\documentclass[11pt]{article}
\usepackage{fullpage}
\usepackage{amsmath}
\usepackage{amssymb}
\usepackage{algorithmic} 
 \usepackage{hyperref}      
 \usepackage{enumerate}
\usepackage{color}
\usepackage{soul}
\usepackage{ifthen}

\newcommand{\algolabel}[1]{\newcounter{#1} \setcounter{#1}{\value{ALC@line}}}
\newcommand{\algoref}[1]{\arabic{#1}}

\newenvironment{proof}{\noindent{\bf Proof.}}{\hfill$\Box$\medskip}
\newenvironment{proofsketch}{\noindent{\bf Proof sketch.}}{\hfill$\Box$}
\newtheorem{theorem}{Theorem}
\newtheorem{lemma}{Lemma}
\newtheorem{definition}{Definition}

\newcommand{\T}{\hspace{0.3cm}}   
\newcommand{\TT}{\hspace*{2em}}
\newcommand{\TTT}{\hspace*{3em}}

\newcommand{\co}[1]{\emph{/* #1 */}}

 \newcommand{\act}[1]{\mbox{#1}}
\newcommand{\tasktuple}[3]{\langle #1,#2,#3 \rangle}
\newcommand{\whp}{\textit{whp}}

\newcommand{\remove}[1]{}
\newtheorem{assert}{Assertion}[section]

\newcommand{\ENed}{enlightened}

\newcommand{\p}[1]{\smallskip\noindent{\bf #1}}

\newcommand{\polylog}[1]{{\it poly}(\log #1)}

\newcommand{\M}[1]{${\cal F}\!\!_{#1}$}

\newcommand{\alg}[1]{{\sc daks}$ #1 $}
\newcommand{\epoch}[1]{epoch~$\!\mathfrak{#1}$}
\newcommand*{\textlabel}[2]{%
  \edef\@currentlabel{#1}
  \phantomsection
  #1\label{#2}
}
\newcommand{\const}[1]{\mathfrak{#1}}

\begin{document}


\title{Technical Report:
Dealing with Undependable Workers\\ in Decentralized Network Supercomputing\thanks{
This work is supported in part by the NSF award 1017232.}
}

\author{
  Seda Davtyan$^*$%
 	 \thanks{$^*$
 	      Department of Computer Science \& Engineering,
         University of Connecticut,
         371 Fairfield Way, Unit 4155,
 	      Storrs CT 06269, USA.
         Emails: {\tt \{seda,acr,aas\}@engr.uconn.edu}.
  } 
 ~ ~ Kishori M. Konwar$^\dag$%
  \thanks{$^\dag$
           University of British Columbia,
           Vancouver, BC V6T 1Z3, CANADA
           Email: {\tt kishori@interchange.ubc.ca}.
 }
~ ~ Alexander Russell$^*$
 ~ ~ Alexander A. Shvartsman$^*$
  }

\date{}


\thispagestyle{empty}
\maketitle

\begin{abstract}
Internet supercomputing is an approach to solving partitionable,
computation-intensive problems by harnessing the power of
a vast number of interconnected computers.
This paper presents a new algorithm
for the problem of using network supercomputing to perform
a large collection of independent tasks, 
while dealing with undependable processors.
The adversary may cause the processors to return
bogus results for tasks with certain probabilities, and may cause
a subset $F$ of the initial set of processors $P$ to crash.
The adversary is constrained in two ways. 
First, for the set of non-crashed processors $P-F$, the \emph{average} probability 
of a processor returning a bogus result is inferior to $\frac{1}{2}$.
Second, the adversary may crash a subset of processors $F$,
provided the size of $P-F$ is bounded from below.
We consider two models: the first bounds the size of $P-F$
by a fractional polynomial, the second bounds this size
by a poly-logarithm.
Both models yield adversaries that are much stronger than previously studied.
Our randomized synchronous algorithm is formulated for $n$ processors
and $t$ tasks, with $n\le t$, where
depending on the number of crashes each
live processor is able to terminate dynamically
with the knowledge that the problem is solved with high probability. 
For the adversary constrained  by a fractional polynomial, the round complexity
of the algorithm is
$O(\frac{t}{n^\varepsilon}\log{n}\log{\log{n}})$, its work is
$O(t\log{n} \log{\log{n}})$ and message complexity is
$O(n\log{n}\log{\log{n}})$.
For the poly-log constrained adversary, the round complexity is
$O(t)$, work is 
$O(t n^{\varepsilon})$,
and message complexity is 
$O(n^{1+\varepsilon})$
All bounds are shown to hold with high probability.
\end{abstract}



\section{Introduction}\label{intro}

Cooperative network supercomputing is becoming increasingly popular for harnessing the power
of the global Internet computing platform. A typical Internet supercomputer, e.g., \cite{DISNET,SETIhome}, 
consists of a master computer  and a large number of computers called workers, performing 
computation on behalf of the master. Despite the simplicity and benefits of a single master approach, 
as the scale of such computing environments grows, it becomes unrealistic to assume the existence 
of the infallible master that is able to coordinate the activities of multitudes of workers. 
Large-scale distributed systems are inherently dynamic and are subject to perturbations, such as 
failures of computers and network links, thus it is also necessary to consider fully distributed peer-to-peer solutions.

One could address the single point of failure issue by providing
redundant multiple masters, yet this would remain a centralized scheme
that is not suitable for big data processing
that involves a large amount of input and output data.
For example, consider  applications
in molecular biology that require 
large reference databases of gene models or annotated protein sequences,
and large sets of unknown protein sequences~\cite{HKWH2014}. 
Dealing with such voluminous 
data 
requires a large scale platform
providing the necessary computational
power and storage. 

Therefore, a more scalable approach is to use a decentralized system, where the input
is distributed and, once the processing is complete, 
the output is distributed across multiple nodes. 
Interestingly, computers returning bogus results is a phenomenon 
of increasing concern. While this may  occur unintentionally,
e.g., as a result of over-clocked processors, workers may in fact
wrongly claim to have performed 
assigned work so as to obtain incentives associated with the system,
e.g.,  higher rank.
To address this problem,
several works, e.g., 
\cite{NCA2011,FGLS2005b,FGLS2005,KRS2006}, study
approaches based on a reliable master 
coordinating unreliable workers. 
The drawback in these approaches is the reliance  
on a reliable, bandwidth-unlimited master processor.

In our recent work \cite{DKS2011,DKS2012} we began to address
this drawback of centralized systems by removing the assumption of 
an infallible and bandwidth-unlimited master processor.
We introduced a decentralized approach, where
a collection of worker processors cooperates on a large
set of independent tasks without the reliance on
a centralized control. Our prior algorithm is able to perform all tasks
with high probability (\whp{}), while dealing with undependable processors under 
an assumption that the average probability of live 
(non-crashed) processors returning incorrect results remains 
inferior to $\frac{1}{2}$ during the computation. 
There the adversary is only allowed to crash a constant fraction of 
processors, and the correct termination of the $n$-processor algorithm 
strongly depends on the availability of $\Omega(n)$ live processors.

The goal of this work is to develop a new $n$-processor algorithm that is able
to deal with much stronger adversaries, e.g., those that can crash
all but a fractional polynomial in $n$, or even a poly-log in $n$, 
number of processors, while still remaining in the synchronous
setting with reliable communication.
One of the challenges here is to enable an algorithm to terminate
efficiently in the presence of any allowable number of crashes.
Of course, to be interesting, such a solution
must be efficient in terms of its work and communication
complexities.

\medskip

\p{{\bf Contributions.}}
We consider the problem of performing $t$ tasks in a distributed 
system of $n$ workers \emph{without} centralized control. 
The tasks are independent, they admit
at-least-once execution semantics, and each task can be
performed by any worker in constant time.
We assume that tasks can be obtained from
some repository (else we can assume that the
tasks are initially known). 
The fully-connected message-passing system is synchronous and
reliable. 
We deal with failure models where crash-prone workers can return
incorrect results.
We present a randomized decentralized algorithm
and analyze it for two different adversaries of
increasing strength: 
constrained by a fractional polynomial, and poly-log constrained.
In each of these settings, we assume that at any point of the
computation 
live processors return bogus results with
the average probability inferior to $\frac{1}{2}$.
In more detail our contributions are as follows.

\begin{itemize}
\item[{\bf 1.}]
Given the initial set of processors $P$, with $|P|=n$,
we formulate two adversarial models,
where the adversary can crash a set $F$ of processors,
subject to the model constraints:
\begin{itemize}
\item[$a)$]
For the first adversary, constrained by a fractional polynomial,
we have $|P-F| = \Omega(n^{\varepsilon})$, for a constant $\varepsilon \in (0,1)$.
\item[$b)$]
For the second, poly-log constrained model, we have
$|P-F| = \Omega(\log^c n)$, for a constant $c \geq 1$. 
\end{itemize}
In both models the adversary may assign arbitrary constant probability
to processors, provided that processors in $P$ return bogus results with the
\emph{average} probability inferior to $\frac{1}{2}$. 
The adversary is additionally constrained, so that
the \emph{average} probability of returning bogus results
for processors in $P-F$ must remain inferior to $\frac{1}{2}$.

\item[{\bf 2.}]
We present a randomized algorithm for $n$ processors and $t$ tasks
that works in synchronous rounds, where each
processor performs a random task
and shares its cumulative knowledge of results with \emph{one} randomly 
chosen processor.
Each processor starts as a ``worker," and
once a processor accumulates a ``sufficient" number of results,
it becomes ``enlightened."
Enlightened processors then ``profess" their knowledge by multicasting
it to exponentially growing random subsets of processors.
When a processor receives a  ``sufficient" number of such
messages, it halts.
We note that workers become enlightened without any synchronization,
and using only the local knowledge. 
The values that control ``sufficient" numbers of results and messages
are established in our analysis and are 
used as \emph{compile-time} constants.

We consider the protocol, by which the ``enlightened" processors ``profess"
their knowledge and reach termination, to be of independent interest.
The protocol's message complexity does not depend on crashes, and the
processors can terminate without explicit coordination. This addresses
one of the challenges associated with termination when $P-F$ can vary
broadly in both models.

\item[{\bf 3.}]
We analyze the quality and performance of the algorithm for the two
adversarial models.
For each model we show that all live workers obtain the
results of all tasks \whp{}, and that these results are
correct \whp{}.
Complexity results for the algorithm also hold \whp{}:
\begin{enumerate}
\item[$a)$] For the polynomially constrained adversary
we show that the algorithm has work complexity   
$O(t\log{n}\log{\log{n}})$ and message complexity  
$O(n\log{n}\log{\log{n}})$. 
\item[$b)$] For the poly-log constrained 
adversary we show that the algorithm has work complexity
$O(t n^\varepsilon)$
and message complexity 
$O(n^{1+\varepsilon})$,
for any $0<\varepsilon <1$.
For this model we note that trivial solutions with all workers doing all tasks
may be work-efficient, but they do not guarantee
that the results are correct.
\end{enumerate}

\end{itemize}

\p{Prior work.}
Earlier approaches explored ways of improving the quality of 
the results obtained from untrusted workers in the settings where
a bandwidth-unlimited and infallible master is coordinating
the workers. Fernandez et al.~\cite{FGLS2005,FGLS2005b} and 
Konwar et al.~\cite{KRS2006} consider a distributed 
master-worker system 
where the workers may 
act maliciously by returning wrong results. 
Works~\cite{FGLS2005,FGLS2005b,KRS2006} design
 algorithms that help the master determine correct results \whp{},
while minimizing work.
The failure models assume that some fraction
of processors can exhibit faulty behavior.
Another recent work by Christoforou et al.~\cite{NCA2011} pursues
a game-theoretic approach.
Paquette and Pelc~\cite{PP2004} consider a model of a 
fault-prone system in which a decision has to be made on the basis of 
unreliable information and design a deterministic  strategy 
that leads to a correct decision \whp{}.

As already mentioned, 
our prior work \cite{DKS2011} introduced the decentralized approach that
eliminates the master, and provided
a synchronous algorithm that is able to perform all tasks
\emph{whp}.
That algorithm requires $\Omega(n)$ live processors to terminate correctly.
Our new algorithm uses a similar approach
to performing tasks, however, it takes a completely different approach
to termination that enables it to tolerate a much broader spectrum of crashes.
The approach uses the new notion of
``enlightened" processors that, having acquired sufficient knowledge,
``profess" this knowledge to other processors,
ultimately leading to termination.
The behavior of two algorithms is similar while $\Omega(n)$
processors remain in the computation, and we use the results
from our prior analysis for this case.

\label{ref:do-all}
A related problem, called Do-All, deals with the setting where a set
of processors must perform a collection of tasks 
in the presence of adversity~\cite{GS2011, KS1997}.
For Do-All the termination condition is that all tasks must
be performed and at least one processor is aware of the fact.
The problem in this paper is different in that each non-crashed
processor must learn the results of all tasks.
Additionally, the failure model in our problem allows processors
to return incorrect results, and our solution requires that
each task is performed a certain minimum number of times
so that the correct result can be discerned, whereas Do-All
algorithms only guarantee that each task is performed at least once.
Thus major changes are required to adapt a solution
for Do-All to our setting.
Do-All, being a key problem in the study of cooperative distributed
computation was considered in a variety of models, 
including message-passing~\cite{DHW1992,PMY1994}
and shared-memory models~\cite{KPRS1991,MS1994}.
Chlebus et al.~\cite{ChlebusGKS02}, study the Do-All problem, in the 
synchronous setting, considering \emph{work} (total number of steps taken) and \emph{communication}
 (total number of point-to-point messages) as equivalent, i.e. they consider the complexity
\emph{work} + \emph{communication} as the cost metric. They derive upper bounds for the bounded adversary,
 upper and lower bounds for $f-$bounded adversary, and almost matching upper and lower bounds for the linearly-bounded 
adversary.

Another related problem is the Omni-Do problem~\cite{DSS2006,GS2003,GRS2005}.
Here the problem is to perform all tasks in a network that
is prone to fragmentations (partitions), thus here the results also must be known
to all processors. However the failure models are quite different
(network fragmentation and merges), and so is the analysis for these models.
For linearly-bounded weakly-adaptive adversaries Chlebus and Kowalski~\cite{CK2004}
give a very efficient randomized algorithm for $t=n$ with work and communication complexities
$O(n \log^* n)$.

Probabilistic quantification of trustworthiness  of participants 
is also used in distributed reputation mechanisms.  
Yu et al.~\cite{YuS02a} propose a probabilistic model for distributed reputation management, 
 where \emph{agents}, i.e., participants, keep ratings of trustworthiness on one another, 
and update the values, using
\emph{referrals} and \emph{testimonies}, through interactions.

\p{Document structure.}   
In Section~\ref{model} we give the models of computation and adversity,
and measures of efficiency. Section~\ref{algorithm} presents our algorithm.
In Section~\ref{analysis} we carry out the analysis of the algorithm and derive  
complexity bounds. 
We conclude in Section~\ref{conclusion} with a discussion.


\section{Model of Computation and Definitions} \label{model}

\p{System model.}
There are $n$ processors, each with a unique 
identifier (id) from  set  $P = [n]$. We refer to the processor 
with id $i$ as processor $i$.
The system is synchronous and processors communicate by exchanging 
reliable messages.
Computation is structured in terms of synchronous \emph{rounds},
where in each round a processor performs three steps:
\emph{send}, \emph{receive}, and \emph{compute}.
In these steps, respectively, processors
can send and receive messages,
and perform local polynomial computation, where the local
computation time is assumed to be negligible compared to message latency. 
Messages received by a processor in a given step include all 
messages sent to it in the previous step. 
The duration of each round
depends on the algorithm.

\p{Tasks.}
There are $t$ tasks to be performed, each with a unique id from  
set ${\cal T} = [t]$. We refer to the task with id $j$ as $Task[j]$.
The tasks are $(a)$~similar, meaning that any task can be done in 
constant time by any processor, ($b$)~independent, meaning that each task 
can be performed independently of other tasks, and (c) idempotent, meaning 
that each task admits at-least-once execution semantics and can be performed
concurrently. 
For simplicity, we assume that the outcome of each task is a binary value.
The problem is most interesting when there are at least as many
tasks as there are processors, thus we consider $t \geq n$.

\p{Models of Adversity.}\label{ref:crash}
Processors are undependable in that a processor may
compute the results of tasks incorrectly and it may crash.
A processor can crash at any moment during the computation;
following a crash, a processor
performs no further actions.

Otherwise, each processor adheres to the protocol
established by the algorithm it executes.
We refer to non-crashed processors as \emph{live}.
We consider an oblivious adversary
that decides prior to the computation what processors to crash and
when to crash them. The maximum number of processors
that can crash is established by the adversarial models
(specified below).

For each processor $i\in P$, we define $p_i$ to be the probability of
processor $i$ returning incorrect results, independently of other
processors, such that, $ \frac{1}{n} \sum_i p_i < \frac{1}{2}-\zeta ,
\textrm{~for~some~} \zeta >0$.
That is, the average probability of processors in $P$
returning incorrect results is inferior to $\frac{1}{2}$.
We use the constant $\zeta$ to ensure that the average
probability of incorrect computation 
does not become arbitrarily close to $\frac{1}{2}$
as $n$ grows arbitrarily large. The individual probabilities
of incorrect computation are unknown to the processors.

For an execution of an algorithm, let
$F$ be the set of processors that adversary crashes. 
The adversary is constrained in that 
the \emph{average} probability of processors in $P-F$ computing
results incorrectly remains inferior to $\frac{1}{2}$.
We define two adversarial models:

\begin{description}
\item \textbf{{\em Model}} \M{fp}, {\it adversary constrained by a
  fractional polynomial} :\\\TT\TT
{
$|P-F| = \Omega(n^\varepsilon)$, for a constant $\varepsilon\in (0,1)$.
}

\smallskip

\item \textbf{{\em Model}} \M{pl}, {\it poly-log constrained adversary} :\\\TT\TT
{
$|P-F| = \Omega(\log^c n)$, for a constant  $c \ge 1$.  
}
\end{description}

\p{Measures of efficiency.}
We assess the efficiency of algorithms in terms of \emph{time}, \emph{work},
and \emph{message} complexities.
We use the conventional measures of  {time complexity}
and {work complexity}.
{Message complexity} assesses the number of point-to-point messages 
sent during the execution of an algorithm.
Lastly, we use the common definition of {\em an event} $\mathcal{E}$ 
{\em occurring with high probability} (\whp{}) to mean that
$\act{\bf Pr}[\mathcal{E}] = 1 - O(n^{-\alpha})$ for some constant $\alpha > 0$.

\section{Algorithm Description} \label{algorithm}

{We now present our decentralized solution, called algorithm \alg{}
(for Decentralized Algorithm with Knowledge Sharing), that employs
 no master and instead uses a gossip-based approach. We start
by specifying in detail the algorithm for $n$ processors and $t=n$ tasks, then
we generalize it for $t$ tasks, where $t\ge n$.} 

The algorithm is structured in terms of a main loop.
The principal data structures at each processor are two arrays of size linear in $n$:
one accumulates knowledge gathered from the processors, and another 
stores the results.
All processors start as \emph{workers}. In each iteration, any worker
performs one randomly selected task and
sends its knowledge
to just one other randomly selected processor.
When a worker obtains ``enough" knowledge about the tasks performed 
in the system, it computes the final results, stops being a worker, and becomes ``\ENed{}."
Such processors no longer perform tasks, and instead ``profess"
their knowledge to other processors by means of
multicasts to exponentially increasing 
random sets of processors.
The main loop terminates when a certain number of 
messages  is received from enlightened processors.
The pseudocode for  algorithm \alg{} is given in Figure \ref{fig:alg}.
We now give the details.

\begin{figure}[tp]
        \hrule \bigskip
\begin{algorithmic}[1] 
{\small
\STATENO{\bf Procedure} for  processor $i$;
\STATEIND\T {\bf external} $n,$ \co{$n$ is the number of processors and tasks}
\STATEIND\TT\TT\T ~$\const{H, K}$ \co{positive constants}\algolabel{alg:const}
\STATEIND\T $Task[1..n]$ \co{set of tasks}
\STATEIND\T $R_i[1..n]$ {\bf init} $\emptyset^n$ \co{set of collected results}
 \STATEIND\T $Results_i[1..n]$ {\bf init}  $\bot$ \co{array of results}
\STATEIND\T ${\it prof\!\_ctr}$  {\bf init} $0$ \co{number of {\sf profess} messages
  received}
\STATEIND\T $r$  {\bf init} $0$ \co{round number}
\STATEIND\T $\ell$  {\bf init} $0$ \co{number of {\sf profess}
  messages to be sent per iteration}
\STATEIND\T $worker$  {\bf init} {\sf true} \co{indicates whether the processor is still a worker}

\STATEIND\T {\bf while}  ${\it prof\!\_ctr} < \const{H} \log{n}$ {\bf do}  \algolabel{alg:term} 
\STATENO\T\T {\em Send:}
\STATEIND\TT\T {\bf if} $worker$ {\bf then}
\STATEIND\TT\T\T  Let $q$ be a randomly selected processor from $ {P} $ \algolabel{alg:share}
\STATEIND\TT\T\T Send $\langle {\sf share}, R_i[\; ] \rangle$ to
processor $q$ \algolabel{alg:share:2}
\STATEIND\TT\T {\bf else}
\STATEIND\TT\T\T  Let $D$ be a set of $2^{\ell} \log{n}$  randomly selected processors from
${P}$ \algolabel{alg:proc} 
\STATENO\TT\T\T \co{Here the selection is with replacement}
\STATEIND\TT\T\T Send $\langle {\sf profess}, R_i[\; ] \rangle$ to processors in $D$
\STATEIND\TT\T\T $\ell \gets \ell + 1$

\STATENO\T\T {\it Receive:}
\STATEIND\TT\T Let $M$ be the set of received messages
\STATEIND\TT\T ${\it prof\!\_ctr} \gets {\it prof\!\_ctr} + |\{ m : m\in
M \wedge m.type = {\sf profess} \}|$  
\STATEIND\TTT\T {\bf for} all $j \in \mathcal{T}$ {\bf do} 
\STATEIND\TTT\T\TT $R_i[j] \gets R_i[j] \cup (\bigcup_{m\in M} m.R[j])$ \co{update knowledge}

\STATENO\T\T {\it Compute:}
\STATEIND\TT\T $r \gets r+1$
\STATEIND\TT\T {\bf if}  $worker$ {\bf then}
\STATEIND  \algolabel{algo:choose_task}\TT\T\T  Randomly select $j \in
\mathcal{T}$ 
and compute the result ${v_j}$ for $Task[j]$
\STATEIND\TT\T\T $R_i[j] \leftarrow  R_i[j] \cup \{\langle v_j,i,r\rangle\}$
\STATEIND\TT\T\T {\bf if} $\min_{j \in \mathcal{T}}\{ |R_i[j]| \} \ge
\const{K} \log{n}$ {\bf then} \algolabel{algo:enled}
\co{$i$ has enough results}
\STATEIND\TT\T\T\T {\bf for  all} $j \in \mathcal{T}$ {\bf do}
\STATEIND\algolabel{algo:majority}
                 \TT\T\TT\T  $Results_i[j] \leftarrow u$ such that  triples $\langle u, \_, \_ \rangle$ form a plurality in $R_i[j]$
\STATEIND\TT\T\T\T $worker \gets {\sf false}$  \co{worker becomes \ENed{}}

\STATENO{\bf end}
\smallskip
\bigskip
\hrule
}
\end{algorithmic}
\caption{\rm Algorithm \alg{} for $t=n$; code at processor $i$ for $i \in P$.}
\label{fig:alg} 
\end{figure}

\p{Local knowledge and state.}
The algorithm is parameterized by $n$, the number of processors and tasks,
and by compile-time constants $\const{H}$ and $\const{K}$ that are discussed later
(they emerge from the analysis).
Every processor $i$ maintains the following:

\begin{itemize}
	\item Array of results $R_i[1..n]$, where 
         element $R_i[j]$, for $j\in \mathcal{T}$, is a 
         set of results for $Task[j]$.  Each $R_i[j]$ is a set of
         triples $\tasktuple{v}{i}{r}$, where
	 $v$ is the result computed for $Task[j]$ by processor $i$ in round $r$ (here the inclusion of  $r$
         ensures that the results computed by processor $i$ in different rounds are preserved).

	\item The array $Results_i[1..n]$ stores the final results. 

	\item The ${\it prof\_ctr}$ stores the number of  messages
          received from enlightened processors.

         \item $r$ is the round (iteration) number that is
            used by \emph{workers} to timestamp the computed results.

          \item $\ell$ is the exponent that controls the number of messages multicast
            by  enlightened processors.
\end{itemize}

\p{Control flow.}
The algorithm iterations are controlled by the main while-loop, and we use the term \emph{round}
to refer to a single iteration of the loop. The loop contains three
stages, viz., \emph{Send}, \emph{Receive}, and
\emph{Compute}. 

Processors communicate using messages $m$
that contain pairs $\langle type, R[\;] \rangle$.
Here $m.R[\;]$ is the sender's array of results.
When a processor is a worker, it sends messages
with $m.type = {\sf share}$. When a processor
becomes enlightened, it  sends messages
with $m.type = {\sf profess}$.
The loop is controlled by the counter ${\it prof\_ctr}$ that keeps
track of the received messages of type {\sf profess}.
We next describe the stages in  detail.

\begin{description}
\item \textbf{\emph{Send} stage:}
Any {worker} chooses a
target processor $q$ at random 
and sends its array of results $R[\;]$ to processor $q$
in a {\sf share} message. 
Any enlightened processor chooses a set $D \subseteq P$ of processors at
random and sends the array of results $R[\;]$ to processors in
$D$ in a {\sf profess} message.
The size of the set $D$ is $2^{\ell} \log{n}$, where initially $\ell
= 0$, and once a processor is enlightened, it increments $\ell$ by $1$
in every round.
(Strictly speaking, $D$ is a multiset, because the random selection
is with replacement. However this is done only
for the purpose of the analysis, and $D$ can be
safely treated
as a set for the purpose of sending  {\sf profess} messages.)
\label{ref:D}

\smallskip

\item \textbf{\emph{Receive} stage:}
Processor $i$ receives messages (if any)
sent to it in the preceding {\it Send} stage.
The processor increments its ${\it prof\_ctr}$ by the
number of {\sf profess} messages received.
For each task $j$, the processor updates its 
$R_i[j]$ by including the results received in all messages.

\smallskip

\item \textbf{\emph{Compute} stage:}
Any worker $i$ randomly selects task $j$, computes the result $v_j$, 
and adds the triple $\langle v_j, i, r \rangle$ for round $r$
to $R_i[j]$. 
For each task the worker checks whether ``enough'' results were collected.
Once at least $\const{K} \log n$ results for each task are obtained, the
worker stores the final results in $Results_i[\;]$ 
by taking the plurality of results for each task, and becomes enlightened.
(In Section~\ref{analysis} we reason about the compile-time
constant $\const{K}$, and establish that $\const{K} \log n$ results are sufficient for our claims.)
Enlightened processors rest on their laurels in subsequent 
\emph{Compute} stages.
\end{description}

\p{{\bf Reaching Termination.}} We note that a processor must become
\ENed{} before it can terminate. Processors can become enlightened at different times
and without any synchronization.
Once enlightened, they profess their knowledge by multicasting
it to exponentially growing random subsets $D$ of processors.
When a processor receives sufficiently many such messages, i.e.,~$\const{H} \log n$, 
it halts, again without any synchronization,
and using only the local knowledge. 
We consider this protocol to be of independent interest.
In Section~\ref{analysis} we reason about the compile-time
constant $\const{H}$, and establish that~$\const{H} \log n$ {\sf profess}
messages are sufficient for our claims; additionally we show that
the protocol's efficiency can be assessed independently of the number of crashes.

\p{Extending the Algorithm for $t\ge n$.}\label{ref:tn}
We now show how to modify the algorithm to handle arbitrary number of tasks $t$
such that $t\ge n$.
Let ${\cal T}' = [t]$ be the set of unique task identifiers, where $t\ge n$. 
We segment the $t$ tasks into chunks of $\lceil t/n \rceil$ tasks, and construct
a new array of chunk-tasks with identifiers in ${\cal T} = [n]$, where
each chunk-task takes $\Theta(t/n)$ time to perform by any live processor.
We now use algorithm \alg{}, where the only difference is that
each \emph{Compute} stage takes  $\Theta(t/n)$ time to perform a chunk-task. 
In the sequel, we use \alg{} as the name of the algorithm when $t=n$, and
we use \alg{_{t,n}} as the name of the algorithm when $t\ge n$.


\section{Algorithm Analysis} \label{analysis}

We present the performance analysis of algorithm \alg{} in the two
adversarial failure models.
We first present the analysis that deals with the case when $t=n$, then
extend the results to the general case with $t\ge n$ for algorithm \alg{_{t,n}}.

\subsection{Foundational Lemmas}\label{sec:lemmas}

We proceed by giving  lemmas relevant to both adversarial models,
starting with the  statement of the well known Chernoff bound.

 \begin{lemma}[{\rm Chernoff Bounds}]\label{chernoff}
 Let $X_1,X_2, ..., X_n$ be $n$ independent
 Bernoulli random variables with $\act{\bf Pr}[X_i=1] = p_i$ and
 $\act{\bf Pr}[X_i=0] = 1-p_i$,
 then it holds for $X=\sum_{i=1}^{n}X_i$
 and $\mu = \mathbb{E}[X] = \sum_{i=1}^{n}p_i$ that
 for all $\delta >0$,
 (i)~$\act{\bf Pr}[X \geq (1+\delta)\mu] \leq e^{-\frac{\mu\delta^2}{3}}$,
 and
 (ii)~$\act{\bf Pr}[X \leq (1-\delta)\mu] \leq e^{-\frac{\mu\delta^2}{2}}$.
 \end{lemma}

Now we show that if $\Theta(n \log{n})$ {\sf profess} messages are sent 
by the \ENed{}
processors, then 
the algorithm
terminates 
\whp{} in one round.

\begin{lemma} \label{lem:termin}
Let $r$ be the first round by which the total number of {\sf profess} messages
is $\Theta(n \log{n})$. Then by the end of this round
every live processor halts \whp{}.
\end{lemma}

\begin{proof}
Let $\tilde{n} =  kn \log{n}$ be the number of
${\sf profess}$ messages sent by round $r$, where $k > 1$ is a sufficiently large
constant. We show that \whp{} every live processor 
receives at least $(1-\delta)k\log{n}$ ${\sf profess}$ messages, 
for some constant $\delta\in (0,1)$.
Let us assume that there exists processor $q$ that receives less than
$(1-\delta)k\log{n}$ of such messages. 
We prove that \whp{} such a processor does not exist.

Since $\tilde{n}$ ${\sf profess}$ messages are sent by round $r$,
there were $\tilde{n}$ random selections of processors from set $P$
in line~\algoref{alg:proc} of  algorithm \alg{} on page~\pageref{fig:alg},
possibly by different \ENed{} processors. We denote by $i$ an
index of one of the random selections in line~\algoref{alg:proc}. 
Let $X_i$ be a Bernoulli random variable such that $X_i = 1$ if
processor $q$ was chosen by an \ENed{} processor and
$X_i=0$ otherwise. 

We define a random variable $X= \sum_{i=1}^{\tilde{n}}X_i$\label{ref:X}
to estimate the total number of times processor~$q$ is
selected by round $r$. 
In line~\algoref{alg:proc}
every \ENed{} processor chooses {a set of destinations}
for the ${\sf profess}$ message
uniformly at random, and hence 
${\bf Pr}[X_i=1] = \frac{1}{n}$. 
Let 
$\mu = \mathbb{E}[X] = \sum_{i=1}^{\tilde{n}}X_i = \frac{1}{n}k
\, n \log{n} = k \log{n}$, then by applying Chernoff bound, for the same
$\delta$ chosen as above, we have:
\begin{displaymath}
{\bf Pr}[ X \leq (1-\delta)\mu ] \leq e^{ -\frac{\mu\delta^2}{2}} 
\leq e^{-\frac{(k\log{n})\delta^2}{2}} \leq \frac{1}{n^{\frac{b\delta^2}{2}}} \leq \frac{1}{n^{\alpha}}
\end{displaymath}

\noindent
where $\alpha >1$ for some sufficiently large $b$. 
We now define $\const{H}$ to be $\const{H} =   (1-\delta)k$.
Thus, with this $\const{H}$, we have ${\bf Pr}[ X \leq \const{H}\log{n}] \leq
\frac{1}{n^{\alpha}}$ for some $\alpha > 1$.
Now let us denote by $\mathcal{E}_q$ the fact that 
 ${\it prof\_ctr}_q \geq \const{H}\log{n}$  by the end of round $r$, and 
let $\bar{\mathcal{E}_q}$ be the complement of that event.  
By Boole's inequality we have 
${\bf Pr}[ \cup_q \bar{\mathcal{E}}_q ] \leq  \sum_q{\bf
  Pr}[\bar{\mathcal{E}_q} ] \leq \frac{1}{n^{\beta}}$, 
where $\beta =\alpha - 1 > 0$. Hence each processor $q \in P$
is the destination of at least $\const{H}\log{n}$ ${\sf profess}$ messages 
\whp{}, i.e., 
$$
{\bf Pr}[ \cap_q \mathcal{E}_q] = 
{\bf Pr}[ \overline{\cup_q \bar{\mathcal{E}_q}} ] 
 = 1- {\bf Pr}[ \cup_q \bar{\mathcal{E}_q} ] 
\geq  1 -\frac{1}{n^{\beta}} 
$$
\noindent
and hence, it halts (line \algoref{alg:term}).
\end{proof}

We use the constant $\const{H}$ from the proof of Lemma~\ref{lem:termin}
as a compile-time constant in algorithm~\alg{} (Figure \ref{fig:alg}).
The constant is used in the main while loop (line~\algoref{alg:term})
to determine when a sufficient number of {\sf profess} messages
is received from enlightened processors, causing the loop to terminate.

We now show that once a processor, that the adversary does not crash,
becomes enlightened then,  \whp{},  in $O(\log{n})$ rounds every 
other live processor becomes enlightened, and 
halts. 

\begin{lemma} \label{lem:halt}
Once a processor $q \in P-F$ becomes \ENed{}, every live
processor halts in additional $O(\log{n})$ rounds \whp{}.
\end{lemma}

\begin{proof}
According to Lemma~\ref{lem:termin} if $\Theta(n \log{n})$ ${\sf
  profess}$ messages 
are sent then every processor halts \whp{}. Given that processor $q$
does not crash 
it takes $q$ at most $\log{n}$ rounds to send $n \log{n}$ ${\sf
  profess}$ messages 
(per line \algoref{alg:proc} in Figure~\ref{fig:alg}), regardless
of the actions of other processors.  
Hence, \whp{} every live processor halts
in $O(\log{n})$ rounds. 
\end{proof}

Next we establish the work and message complexities for
algorithm \alg{} for the case when number of crashes is small, specifically when
at least a linear number of processors do not crash.
As we mentioned in the introduction, while $\Omega(n)$ processors
remain active in the computation, algorithm \alg{} performs tasks
in exactly the same pattern as algorithm $A$ in \cite{DKS2011}
(to avoid a complete restatement, we kindly refer the reader to that earlier paper).
This forms the basis for the next lemma.

\begin{lemma}\label{lem:epocha}
Algorithm \alg{} has work and message complexity 
{$\Theta(n \log n)$}
when $\Omega(n)$ processors do not crash.
\end{lemma}

\begin{proof}
Algorithm \alg{} chooses tasks to perform in the same
 pattern as algorithm  $A$ in \cite{DKS2011}, however
the two algorithms have very different termination strategies.
Theorems~2 and~4 of \cite{DKS2011} establish that
in the presence of at most $f \cdot n$ crashes, for a constant $f\in(0,1)$, 
the work and message
complexities of algorithm $A$ are $\Theta(n \log{n})$.
The termination strategy of algorithm \alg{} is completely different,
however, per Lemmas~\ref{lem:termin} and~\ref{lem:halt}, 
after at least one processor from $P-F$ is \ENed{} every live
  processor halts in $\Theta(\log n)$ rounds  \whp{}, having sent
  $\Theta(n \log n)$ ${\sf profess}$ messages.
Thus, with at least a linear number of processors
remaining, the work and message complexities,
relative to algorithm~$A$, increase by an additive
$\Theta(n \log n)$ term.
The result follows.
\end{proof}

We denote by $L$ the number of rounds required for a processor from the
set $P-F$ to become \ENed{}. We next 
analyze the value of $L$ for models \M{fp} and \M{pl}.


\subsection{Analysis for  Model \M{fp}} \label{model:Fa}

In model \M{fp} we have $|F| \le n - n^\varepsilon$.
Let $F_r$ be the actual number of crashes that occur
prior to round $r$. 
For the purpose of analysis we divide an execution
of the algorithm into two epochs:
\epoch{a} consists of all rounds $r$ where $|F_r|$ is at most linear in $n$, 
so that when the number of live processors is at least $c'n$
for some suitable constant $c'$;
\epoch{b} consists of all rounds $r$ starting with first round $r'$ 
(it can be round 1) when the number of live
processors drops below some $c'n$ and becomes
$c''n^\varepsilon$ for some suitable constant $c''$.
Note that either epoch may be empty.

For the small number of crashes in \epoch{a}, 
Lemma~\ref{lem:epocha} gives the worst case
work and message complexities as $\Theta(n \log n)$;
the upper bounds apply whether or not the algorithm
terminates in this epoch.

Next we consider \epoch{b}.
If the algorithm terminates in round $r'$, the first round of the epoch,
the cost remains the same as
given by Lemma~\ref{lem:epocha}.
If it does not terminate, it incurs additional
costs associated with the processors in $P - F_{r'}$,
where $|P - F_{r'}| \le c'' n^\varepsilon$. 
We analyze the costs for \epoch{b} in the rest
of this section. The final message and work complexities
will be at most the worst case complexity for \epoch{a}
plus the additional costs for \epoch{b}
incurred while $|P-F|= \Omega(n^\varepsilon)$ per model \M{fp}.

First we show that \whp{} it will take $L= O(n^{1-\varepsilon} \log{n}
\log{\log{n}})$ rounds for a worker from the set $P-F$ to become \ENed{}
in \epoch{b}.

\begin{lemma} \label{thresholdFa}
In $O(n^{1-\varepsilon}\log{n})$ rounds of \epoch{b} 
every task is performed
$\Theta(\log{n})$ times \whp{} by processors in $P-F$.
\end{lemma}

\begin{proof}
{If the algorithm terminates within 
$O(n^{1-a}\log{n})$
rounds of \epoch{b}, then each task is performed
$\Theta(\log n)$ times as reasoned earlier.
Suppose the algorithm does not terminate
(in this case its performance is going to be worse).
}

Let us assume that after $\tilde{r}  = \kappa n^{1-a}\log{n}$  
rounds of algorithm \alg{}, where
$\kappa$ is a sufficiently large constant and $0<a <1$ is a constant,
there exists a task $\tau$ that is performed less than $(1-\delta) \kappa \log{n}$ times 
among all live workers, for some $\delta > 0$. 
We prove that \whp{} such a task does not exist.

We define $k_2$ to be such that $k_2 = (1-\delta) \kappa$
(the constant $k_2$ will play a role in establishing the value of the compile-time
  constant $\const{K}$ of algorithm \alg{}; we come
  back to this in Section~\ref{analysis}). According to the above assumption,
 at the end of round $\tilde{r}$ for some task $\tau$, we have 
$|\cup_{j=1}^{n} R_j[\tau]| < (1-\delta) \kappa \log{n} = k_2 \log{n}$.

Let us consider all algorithm iterations individually performed by each processor in $P-F$
during the $\tilde{r}$ rounds.
Let $\xi$ be the total number of such individual iterations. 
Then $\xi \ge \tilde{r} |P-F| \geq cn^a$. 
During any such iteration, a processor from $P-F$ selects and performs 
task $\tau$ in  line~\algoref{algo:choose_task} independently
with probability $\frac{1}{n}$.
Let us arbitrarily enumerate said iterations from $1$ to~$\xi$.
Let $X_1, \ldots, X_x, \dots, X_\xi$ be Bernoulli random variables, such that
$X_x$ is $1$ if task  $\tau$ is performed in iteration $x$, and $0$ otherwise.
We define $X \equiv  \sum_{x=1}^{\xi}X_x$, the random variable that describes
the total number of times task $\tau$ is performed during the 
$\tilde{r}$ rounds by processors in $P-F$. 
We define $\mu$ to be ${\mathbb E}[X]$.
Since ${\bf Pr}[X_x=1] = \frac{1}{n}$, for $x\in \{1, \ldots, \xi\}$,
where $\xi \geq \tilde{r}cn^a$,
by linearity of expectation, we obtain
 $\mu = {\mathbb E}[X] = \sum_{x=1}^{\tilde{r} cn^{a}} \frac{1}{n} =
\frac{\tilde{r} cn^{a}}{n} > \kappa c \log{n} $.
Now by applying  Chernoff bound, for the same $\delta>0$ chosen as above, we have: 
\begin{displaymath}
{\bf Pr}[ X \leq (1-\delta)\mu ] \leq e^{ -\frac{\mu\delta^2}{2}} 
\leq e^{-\frac{(\kappa c\log{n})\delta^2}{2}} \leq \frac{1}{n^{\frac{b\delta^2}{2}}} \leq \frac{1}{n^{\alpha}}
\end{displaymath}

\noindent
where $\alpha >1$ for some sufficiently large $b$.
Now let us denote by $\mathcal{E}_\tau$ the fact that 
 $|\cup_{i=1}^{n} R_i(\tau)| > k_2 \log{n}$ by the round $\tilde{r} $ of the algorithm and  
we denote by $\bar{\mathcal{E}_\tau}$ the complement of that event.  
Next by Boole's inequality we have 
${\bf Pr}[ \cup_\tau \bar{\mathcal{E}}_\tau ] \leq  \sum_\tau{\bf Pr}[\bar{\mathcal{E}_\tau} ] \leq \frac{1}{n^{\beta}}$,
where $\beta =\alpha - 1 > 0$. Hence each task  is performed at least
$\Theta(\log{n})$ times  by workers in $P-F$ \whp{}, i.e., 

\hfill
$\displaystyle
{\bf Pr}[ \cap_\tau \mathcal{E}_\tau ] = 
{\bf Pr}[ \overline{\cup_\tau\bar{\mathcal{E}_\tau}} ] \geq  1 -
\frac{1}{n^{\beta}}.
$
\end{proof}

We now focus only on the set of live processors $P-F$
with  $|P-F| \geq cn^{\varepsilon}$. 
Our goal is to show that in $O(n^{1-\varepsilon} \log{n}
\log{\log{n}})$ rounds of algorithm \alg{} at least one processor from
$P-F$ becomes \ENed{}.
{In reasoning about Lemmas~\ref{lem:dis:log},~\ref{lem:dissFa-p1}
  and~\ref{lem:dissFa-p2}, that follow, we note that
if the algorithm terminates within 
$O(n^{1-\varepsilon} \log{n} \log{\log{n}})$
rounds of \epoch{b}, then every processor in $P-F$ is enlightened as
reasoned earlier. Suppose the algorithm does not terminate
(in focusing on this case we note that the algorithm's performance is going to be worse).}

We first show that any triple $z$ generated by a processor in $P-F$ is known to all
processors in $P-F$  in $O(n^{1-\varepsilon} \log{n} \log{\log{n}})$ rounds of
algorithm~\alg{}.

We denote by $S(r) \subseteq P-F$ the set of processors that know a certain triple
$z$ by round $r$, and let $s(r) = |S(r)|$. 
Next lemma shows 
that 
by round $r_1 = r' + \Theta(n^{1-\varepsilon} \log{n} \log{\log{n}})$
in \epoch{b} we have
$s(r_1) = \Theta(\log^3{n})$.

\begin{lemma} \label{lem:dis:log}
By round  $r_1 = r' + \Theta(n^{1-\varepsilon} \log{n} \log{\log{n}})$
of \epoch{b},
$s(r_1) = \Theta(\log^3{n})$ \whp{}.
\end{lemma}

\begin{proof}
Consider a scenario when a processor $p \in P-F$ generates a triple
$z$. Then the probability that processor $p$ sends triple $z$
to at least one other processor $q \in P-F$, where $p \neq
q$, in $n^{1-\varepsilon} \log{n}$ rounds is at least

\hfill 
$1 - (1 - \frac{cn^\varepsilon}{n})^{n^{1-\varepsilon} \log{n}} \geq 1
- e^{-b\log{n}} > 1- \frac{1}{n^{\alpha}}$
s.t. $\alpha > 0$
\hfill ~

\noindent
for some appropriately chosen $b$ and for a
sufficiently large $n$. Similarly, it is straightforward to show that
the number of live processors that learn about $z$ doubles every
$n^{1-\varepsilon} \log{n}$ rounds, hence \whp{} after
$(n^{1-\varepsilon} \log{n})\cdot 3 \log{\log{n}} = 
\Theta(n^{1-\varepsilon} \log{n}\log{\log{n}})$ rounds the number of
processors in $P-F$ that learn about $z$ is 
$\Theta(\log^3{n})$. 
\end{proof}

In the next lemma we reason about the growth of $s(r)$ after round
{$r_1$}.

\begin{lemma} \label{lem:dissFa-p1}
Let $r_2$ be the first round after round {$r_1$} in \epoch{b}
such that  $r_2-r_1 = \Theta(n^{1-\varepsilon} \log{n})$.
Then $s(r_2) \geq \frac{3}{5}|P-F|$ \whp{}.
\end{lemma}

\begin{proof}
Per model \M{fp}, let constant $c$ be such that
$|P-F| \geq cn^{a}$.
We would like to apply the Chernoff bound to approximate the number of
processors from 
  $(P-F)-S(r_1)$ that learn about triple $z$ by round $r_2$. 
According to algorithm \alg{} if a processor $i \in (P-F)-S(r_1)$ 
learns about triple $z$ in some round $r_1<r<r_2$, then in round $r+1$
processor $i$  forwards $z$ to some randomly chosen
processor $j \in P$ (lines~\algoref{alg:share}-\algoref{alg:share:2} of the algorithm). Let $Y_i$, where $i \in
(P-F)-S(r_1)$, be a random variable such that $Y_i = 1$ if processor
$i$ receives the triple $z$ from some processor $j \in S(r_1)$, in
some round  $r_1<r<r_2$, and $Y_i=0$ otherwise. It is clear that if
some processor $k \in (P-F)-S(r)$, where $k \neq i$ receives triple
$z$ from processor $i$ in round $r+1 < r_2$, then random variables
$Y_i$ and $Y_k$ are not independent, and hence, the Chernoff bound cannot
be applied. To circumvent this,
we consider the rounds between $r_1$ and $r_2$ and partition
these rounds into blocks of $\frac{1}{c}n^{1-a}$ consecutive rounds. 
For instance, rounds $r_1+1,...,r_1+\frac{1}{c}n^{1-a}$ form
the first block, rounds
$r_1+\frac{1}{c}n^{1-a}+1,...,r_1+\frac{2}{c}n^{1-a}$ form the
second block, etc. The final block may contain less than
$\frac{1}{c}n^{1-a}$ rounds. 

We are interested in estimating the
fraction of the processors in $(P-F) - S(r_1)$ that learn about triple $z$ at
the end of each block. 

For the purpose of the analysis we consider another algorithm,
called \alg{'}. The difference between
 algorithms \alg{} and \alg{'} is that in \alg{'} a processor does not forward
 triple $z$ in round $r$ if $z$ was first received in the round
 that belongs to the same block as $r$ does.
This allows us to apply Chernoff bound (with negative
dependencies) to approximate the number of processors in
  $(P-F)-S(r_1)$ that learn about triple $z$ in a block.
We let $S^{\prime}(r)$ be the subset of processors in $P-F$
that are aware of triple $z$ by round $r$ in algorithm \alg{'}, and we let $s'(r)=|S'(r)|$.
Note, that since in \alg{'} triple $z$ is forwarded less
  often than in \alg{}, it follows that the number of processors from
  $P-F$ that learn about $z$ in \alg{} is at least as large as the
  number of processors from $P-F$ that learn about $z$ in \alg{'},
and, in particular, $S'(r)\subseteq S(r)$, for any $r$.
This allows us to consider algorithm \alg{'}
  instead of \alg{} for assessing the number of processors from
  $P-F$ that learn about $z$ by round $r_2$,
and we do this by having $s'(r)$ serve as a lower bound for $s(r)$.

Let $X_i$, where $i \in (P-F) - S^{\prime}(r)$, be a random
variable, s.t. $X_i = 1$ if processor $i$ receives the triple $z$ from some
 processor $j \in S^{\prime}(r)$ in a block that starts with round
 $r+1$, e.g., for the first block $r=r_1,$ and $X_i=0$ otherwise. 
Let us next define the random variable $X= {\sum}_{i \in (P-F) -
  S^{\prime}(r)}{X_i} = s^{\prime}({r+\frac{1}{c}n^{1-a})}-
s^{\prime}(r)$  
to count the number of processors in $(P-F)-S^{\prime}(r)$ that received
triple $z$ in the block that starts with round $r+1$. 

Next, we calculate $\mathbb{E}[X]$, the expected number
of processors in $(P-F) - S^{\prime}(r)$ that learn about triple $z$ at the
end of the block that begins with round $r+1$ in algorithm \alg{'}.
There are $s^{\prime}(r)$ processors in $S^{\prime}(r)$ that are aware of
triple $z$. 
Note that there are $\frac{1}{c}n^{1-a}$ consecutive rounds in a block;
and during every round every processor $q \in S'(r)$ 
 picks a processor from $P$ uniformly at random, and sends the triple $z$ to it.
Note also that in algorithm \alg{'}, triple $z$ is not forwarded by a processor during the 
same round in which it is received.\label{ref:1n} 
Therefore,  every processor $p$ in  $(P-F) - S^{\prime}(r)$ has a
probability of $\frac{1}{n}$ to be selected by a processor 
$q \in S^{\prime}(r)$ in one round. Conversely, the probability that $p \in
(P-F) - S^{\prime}(r)$ is not selected by $q$ is $1-\frac{1}{n}$. 
The number of trials is $\frac{s^{\prime}(r)}{c}
n^{1-a}$, hence the probability that processor $p \in
(P-F) - S^{\prime}(r)$ is not selected is
$(1-\frac{1}{n})^{\frac{s^{\prime}(r)}{c}n^{1-a}}$. On the contrary, the
probability that a processor $p \in (P-F) - S^{\prime}(r)$
is  selected is
$1-(1-\frac{1}{n})^{\frac{s^{\prime}(r)}{c}n^{1-a}}$. Therefore, the
expected number of processors from $(P-F) - S^{\prime}(r)$
that learn about triple $z$ by the end of the block in
algorithm \alg{'} is $\displaystyle
 (cn^{a} - s'(r))(1-
(1-\frac{1}{n})^{\frac{s'(r)}{c}n^{1-a}})  $. 
Next, by applying the binomial expansion, we have:

$\displaystyle
 (cn^{a} - s'(r))(1-
(1-\frac{1}{n})^{\frac{s'(r)}{c}n^{1-a}})  $
\begin{eqnarray*} 
&& \geq   
 (cn^{a}-s'(r))(\frac{s'(r)n^{1-a}}{cn}
- \frac{{s'(r)}^2n^{2-2a}}{2c^2n^2}) \\  
&& = 
cn^{a} (1-\frac{s'(r)}{cn^{a}})
\frac{s'(r)n^{1-a}}{cn}(1-\frac{s'(r)n^{1-a}}{2cn}) \\
&&=
 s'(r) (1-\frac{s'(r)}{cn^{a}}) (1-\frac{s'(r)n^{1-a}}{2cn})
\end{eqnarray*}

The number of processors from $(P-F)- S^{\prime}(r)$ that become aware of
triple $z$ in the block of $\frac{1}{c}n^{1-a}$ rounds that starts
with round $r+1$ is 
$s^{\prime}({r+\frac{1}{c}n^{1-a}})-s^{\prime}(r)$. 
While, as shown above, the expected number of
processors that learn about triple $z$ is 
$\mu = \mathbb{E}[X]=s^{\prime}(r)(1-
\frac{s^{\prime}(r)}{cn^{a}})(1-\frac{s^{\prime}(r)}{2cn^{a}})$.

On the other hand, because in algorithm \alg{'} no processor that learns
about triple $z$ in a block forwards it in the same block, we have
negative dependencies among the random variables $X_i$.
And hence, we can apply the regular
Chernoff bound, with $\delta = \frac{1}{\log{n}}$. 
Considering  also that $s'(r) \geq s(r_1)$ and that $s(r_1) = \Theta(\log^3n)$ by Lemma~\ref{lem:dis:log}, 
we obtain:
\begin{eqnarray*}
{\bf Pr}[X \leq (1-\frac{1}{\log{n}})\mu] \leq
e^{-s^{\prime}(r)(1-\frac{s^{\prime}(r)}{cn^{a}})(1-\frac{s^{\prime}(r)}{2cn^{a}})\frac{1}{2}\frac{1}{\log^2{n}}}
\leq e^{-\frac{k\log^3{n}}{2\log^2{n}}} = e^{-k\log{n}} 
\leq \frac{1}{n^\alpha}
\end{eqnarray*}

where $\alpha > 0$ for some sufficiently large $k>2$.

Therefore, \whp{}  the number of processors that learn about
  triple $z$ in a block
  that starts with round $r+1$ is

\begin{eqnarray*}
s'({r+\frac{1}{c}n^{1-a}}) & \geq &  s'(r) +
s'(r) (1-\frac{ s'(r)}{cn^{a}}) (1 -
\frac{s'(r)}{2cn^{a}}) (1 - \frac{1}{\log{n}}) \\
& \geq& s'(r) +
s'(r) (1-\frac{ 3s'(r)}{2cn^{a}}) (1 - \frac{1}{\log{n}}) \\
&\geq & s'(r) + s'(r) (1 - \frac{3}{2}\cdot\frac{3}{5})(1 -
\frac{1}{\log{n}}) \\
&\geq& s'(r) \frac{21}{20}
\end{eqnarray*}
for a sufficiently large $n$, and given that $s'(r) <
  \frac{3}{5}cn^a$ (otherwise the lemma is proved). 

Hence we showed that the number of processors from
  $(P-F)-S'(r)$ that learnt about triple $z$ at the end of the block
that starts with round $r+1$, is at least $\frac{21}{20}s'(r)$
\whp{}. It remains to show that $s(r_2) \geq
\frac{3}{5}|P-F|$ \whp{}. Indeed, even assuming that processors that learnt
about triple $z$ following round $r$ do not disseminate it, after
repeating the process described above for some $O(\log{n})$ times, it
is clear that \whp{} $s^{\prime}({r_2}) \geq \frac{3}{5}|P-F|$. On the
other hand since the block size is $\frac{1}{c}n^{1-a}$ and $r_2 - r_1
= \Theta(n^{1-a} \log{n})$ there are $\Theta(\log{n})$ blocks.

Thus \whp{} we have $s^{\prime}({r_2}) \geq \frac{3}{5}|P-F|$ for
$r_2 - r_1 = \Theta(n^{1-a} \log{n})$, and since $S^{\prime}(r)
\subseteq S(r)$ we have $s({r_2}) \geq \frac{3}{5} |P-F|$. 
\end{proof}

In the proof of the next lemma we use  the {\em Coupon Collector's} problem~\cite{RAJEEV95}:

\begin{definition} {The Coupon Collector's Problem (CCP).}
\label{def:ccp} 
There are $n$ types of coupons and at each trial a coupon is chosen at random.
Each random coupon is equally likely to be of any of the n types, and the random
choices of the coupons are mutually independent. Let $m$ be the number of trials.
The goal is to study the relationship between $m$ and the probability of having
collected at least one copy of each of $n$ types.
\end{definition}

In~\cite{RAJEEV95} it is shown that $\mathbb{E}[X]=n \ln{n} + O(n)$ and that \whp{} 
the number of trials for collecting all $n$ coupon
types lies in a small interval centered about its expected value.

Next we calculate the number of rounds required for
the remaining $\frac{2}{5}|P-F|$ processors in $P-F$  to learn  $z$.
Let $U_d \subset P-F$ be the set of workers that do not learn
$z$ after $O(n^{1-\varepsilon} \log{n} \log{\log{n}})$ rounds
of algorithm \alg{}.  According to Lemma~\ref{lem:dissFa-p1} we have $|U_d| \leq
\frac{2}{5}|P-F|$. 

\begin{lemma} \label{lem:dissFa-p2}
Once every task is performed $\Theta(\log{n})$ times in \epoch{b}
by processors in $P-F$
then at least one worker from $P-F$ becomes \ENed{}
in $O(n^{1-\varepsilon} \log{n} \log{\log{n}})$ rounds, \whp{}. 
\end{lemma}

\begin{proof}
According to Lemmas~\ref{lem:dis:log} and~\ref{lem:dissFa-p1} in
$O(n^{1-\varepsilon} \log{n} \log{\log{n}})$ rounds of algorithm \alg{}
at least $\frac{3}{5}|P-F|$ of the workers are aware of triple  
$z$ generated by a processor in $P-F$.  Let us denote this subset of workers by $S_d$,
where $d$ is the first such round.

We are interested in the number of rounds 
required for every processor in $U_d$ to learn about $z$ \whp{} by receiving a 
message from a processor in $S_d$ in some round following $d$.

We show that, by the analysis  similar to CCP,
 in $O(n^{1-\varepsilon} \log{n} \log{\log{n}})$ rounds triple $z$
 is known to all processors in $P-F$, \whp{}.
Every processor in $P-F$ has a unique id, hence we consider these processors
as different types of coupons and we assume that the processors in $S_d$ collectively
represent the coupon collector. In this case, however, we do not require that
every processor in $S_d$ contacts all processors in $U_d$ \whp{}. 
Instead, we require only that
the processors in $S_d$ \emph{collectively} contact all processors in $U_d$ \whp{}.
According to our algorithm in every round every processor in $P-F$ $(S_d \subset P-F)$,
selects a processor uniformly at random and sends all its data 
to it.
Let us denote
by $m$ the collective number of trials by processors in $S_d$ to contact processors in $U_d$.
According to CCP if $m= O(n \ln{n})$ then \whp{} processors in $S_d$ collectively contact every 
processor in $P-F$, including those in $U_d$. Since there are at least $\frac{3}{5} cn^{\varepsilon}$ 
processors in $S_d$ then in every round the number of trials is at least $\frac{3}{5}c n^{\varepsilon}$,
hence in $O(n^{1-\varepsilon} \ln{n})$ rounds  \whp{} all processors in $U_d$ learn about $z$.
Therefore, in $O(n^{1-\varepsilon}\log{n} \log{\log{n}})$ rounds \whp{} all processors in
$U_d$, and thus in $P-F$, learn about $z$. 

Let $\mathcal{V}$ be the set of triples such that for every task $j \in
\cal{T}$ there are $\Theta(\log{n})$ triples generated by processors
in $P-F$, and hence $|\mathcal{V}| = \Theta(n \log{n})$.
Now by applying Boole's inequality we want to show that \whp{} in
$O(n^{1-\varepsilon}\log{n} \log{\log{n}})$ rounds   
all triples in $\cal{V}$ become known to all processors in $P-F$.

Let $\overline{\mathcal{E}}_z$ be the event that some triple $z
\in \cal{V}$ is not known to all processors in $P-F$. In the preceding part of the proof 
we have shown that ${\bf Pr}[\overline{\mathcal{E}}_z] < \frac{1}{n^\beta}$,
where $\beta>1$. By Boole's inequality, the probability that there exists one
 triple in $\cal{V}$ that is not known to all processors in $P-F$ can be bounded as 
\begin{displaymath}
{\bf Pr} [\cup_{z \in {\mathcal
    V}}\overline{\mathcal{E}}_z ] \leq \Sigma_{z \in
  {\mathcal V}}{\bf Pr}[\overline{\mathcal{E}}_z] = 
\Theta(n \log n )\frac{1}{n^\beta} \leq \frac{1}{n^\gamma}
\end{displaymath}
where $\gamma>0$. This implies that 
every  processor in $P-F$ collects all $\Theta(n \log n)$
triples generated by processors in $P-F$, \whp{}. And hence, at least one
of these processors becomes \ENed{} after $O(n^{1-\varepsilon}
\log{n} \log{\log{n}})$ rounds. 
\end{proof}

\begin{theorem}\label{thm:correct:resFa}
Algorithm \alg{} makes known the correct results of all $n$ tasks
at every live processor in \epoch{b}
after $O(n^{1-\varepsilon} \log{n} \log{\log{n}})$ rounds \whp{}.
\end{theorem}

\begin{proof}
According to algorithm \alg{} (line~\algoref{algo:majority}) every
live processor computes the result of every task $\tau$ by taking a
plurality among all the results.
We want to  prove that the majority of the results for any
task $\tau$ are correct at any \ENed{}  processor, \whp{}.

To do that, for a task $\tau$ we estimate (with a concentration
bound) the number of times the results are computed correctly, then
we estimate the bound on the total number of times task $\tau$ is
computed (whether correctly or incorrectly), and we show
that a majority of the  results are computed correctly. 

Let us consider random variables
$X_{ir}$ that denote the success
or failure of correctly computing the result of some task $\tau$ in 
round $r$ by worker $i$. Specifically,  $X_{ir} = 1$ if in 
round $r$, worker $i$ computes the result of the task $\tau$ 
correctly, otherwise $X_{ir}=0$. According to our algorithm we observe
 that for a live processor $i$ we have ${\bf Pr}[X_{ir}=1] = \frac{q_i}{n}$ and 
 ${\bf Pr}[X_{ir}=0] = 1- {\bf Pr}[X_{ir}=1]$, where $q_i \equiv 1 - p_i$. 
We want to count the number of correct results
calculated for task $\tau$ when a processor  $i \in P-F$ becomes
\ENed{}.
As before, we let $F_r$ be the set of processors that crashes
prior to round $r$. 
 
Let $X_r \equiv \sum_{i \in P-F_r}X_{ir}$ denote the number of correctly
computed results for task $\tau$ among all live workers during round $r$.
By linearity of expected values of a sum of random variables we have
\begin{displaymath}
 \mathbb{E}[X_{r}] = \mathbb{E}\left[\sum_{i \in P-F_r}X_{ir}\right] = \sum_{i \in P-F_r}\mathbb{E}[X_{ir}]  
= \sum_{i \in P-F_r}\frac{q_i}{n} 
\end{displaymath}

We denote by $L'$ the minimum number of rounds required for at least
one processor from $P-F$ to become enlightened. 
It follows from line~\algoref{algo:enled} of algorithm \alg{}
that a processor
becomes \ENed{} only when there are at
least $\const{K} \log{n}$ results for every task $\tau \in \cal{T}$
(the constant $\const{K}$ is chosen later in this section).
{We see that  $L' \geq \tilde{c} n^{1 - a}\const{K}
  \log{n}$, where $0 < \tilde{c} \leq 1$. This is because there are $t = n$
tasks to be performed, and  in \epoch{b} we have $|P-F| \geq
  cn^{a}$ for a constant $c \geq 1$}.\label{ref:epoch}

We further denote by $X \equiv \sum_{r=1}^{L'}X_{r}$
the number of correctly computed results for task $\tau$
when the condition in line~\algoref{algo:enled} of the algorithm  is satisfied.
Again, using the linearity of expected values of a sum 
of random variables we have 
\begin{displaymath}
\mathbb{E}[X] = \mathbb{E}\left[\sum_{r=1}^{L'}X_r\right]= \sum_{r=1}^{L'}\mathbb{E}[X_r]= \sum_{r=1}^{L'}\sum_{i \in P-F_r}\frac{q_i}{n} =\frac{L'}{n}\sum_{i \in P-F_r} q_i
\end{displaymath}

Note that, according to our adversarial model definition, for every round $r \leq
L'$  we have $\frac{1}{|P-F_r|}\sum_{i \in P-F_r} q_i > \frac{1}{2}
+\zeta'$,  for some fixed $\zeta' > 0$. 
Note also that
$\frac{1}{|P-F_r|}\sum_{i \in P-F_r} q_i \geq \frac{1}{n}\sum_{i \in P-F_r} q_i $,
and hence, there exists 
some $\delta >0$, such that, $ (1-\delta)\frac{L'}{n}\sum_{i \in P-F_r} q_i > (1+\delta) \frac{L'}{2}$. Also,
observe that the random variables 
 $X_{1}, X_{2}, \ldots, X_{L'}$ are mutually independent, since we
 consider an oblivious adversary and the random variables correspond to different 
rounds of execution of the algorithm.
Therefore, by applying Chernoff bound on $X_{1}, X_{2}, \ldots, X_{L'}$ we have:
$$
 {\bf Pr}\left[ X \leq (1 - \delta)\mathbb{E}[X]\right] = {\bf Pr}\left[X \leq
 (1-\delta)\frac{L'}{n}\sum_{i \in P-F_r}{q_i}\right]  
 \leq e^{-\frac{\delta^2 L'(1+\delta)}{4(1-\delta)}} \leq \frac{1}{n^{{\alpha_1}}} \; ,
$$

\noindent
{where $L' \geq \tilde{c} n^{1 - a}\const{K}\log{n}$ as above
  and ${\alpha_1} > 1$ for a sufficiently large $n$.}

Let us now count the total number of times task $\tau$ is chosen to be performed during the 
execution of the algorithm until every live processor halts.  We represent the choice of 
task $\tau$ by worker $i$ during round $r$ by a random variable $Y_{ir}$.
We assume $Y_{ir}=1$ if $\tau$ is chosen by worker $i$ in round $r$, otherwise $Y_{ir}=0$.
    
   At this juncture, we address a technical point regarding the total number of 
results for $\tau$ used for computing plurality. Note that even after round $L'$
 any processor that is still a worker continues to perform tasks, thereby
 adding more results for  task $\tau$.
According to Lemma~\ref{lem:halt} every
processor is enlightened in $O(\log{n})$ rounds after
$L'$. Furthermore, in \epoch{b} following round $L'$, the number
of processors that are still  workers is  $n'' < |P-F|$. Hence, the expected number of results
computed for every task $\tau$ by workers is  
   $k\frac{\log{n}}{ n^{a}}$, for some $k>0$,
 that is,  $O(\frac{1}{n^{a'}})$, for some $a' > 0$.
 Therefore, the number of results  computed for task $\tau$, starting from round
$L'$ and until the termination is negligible. 
Let us denote by $Y$ the total number of results computed for
a task $\tau$ at  termination. 
We express the random variable $Y$ as 
 $Y \equiv \sum_{r=1}^{L} \sum_{i \in P-F_r}Y_{ir}$,
where $L$ is the last round prior to  termination.
{As argued above, the total number of results computed for task
  $\tau$ between rounds $L'$ and $L$ is $O(\frac{1}{n^{a'}})$, for
  some $a' > 0$, and hence $\frac{L}{L'} = 1 + o(1)$.}
Note that the outer sum terms of $Y$ consisting of the inner sums are mutually independent 
because each sum pertains to a different round; this
 allows us to use Chernoff bounds. 
From above it is clear
that $\mathbb{E}[Y] = \tilde{c} n^{1 - a}\const{K}\log{n} + \frac{1}{n^{o(1)}}$.  
Therefore, by applying Chernoff
bound for the same $\delta >0$ as chosen above we have: 
\begin{displaymath}
  {\bf Pr}[Y \geq (1 + \delta) \mathbb{E}[Y]] =  {\bf Pr}[Y \geq (1 +
  \delta)L'] \leq 
    e^{-\frac{\delta^2 \tilde{c} n^{1 - a}\const{K}\log{n}}{3}} \leq \frac{1}{n^{{\alpha}_2}} \; ,
\end{displaymath}

\noindent
{where ${\alpha}_2 > 1$ for a sufficiently large $n$.}

Then, by applying Boole's inequality  on the above two
 events, we have
\begin{displaymath}
 {\bf Pr}[\{X \leq (1-\delta)\frac{L'}{n}\sum_{i \in P-F_r}{q_i}\} \cup 
   \{Y \geq (1 + \delta)L'\}]  \leq \frac{2}{n^{\alpha}} 
\end{displaymath}

\noindent
where $\alpha=\min\{{\alpha}_1,{\alpha}_2\} > 1$

Therefore,  from above, and from the fact that $
(1-\delta)\frac{L'}{n}\sum_{i \in P-F_r} q_i > (1+\delta)
\frac{L'}{2}$,
we have ${\bf Pr}[ Y/2 < X] \geq  1 - \frac{1}{n^{\beta}}$
 for some $\beta > 1$. Hence, at termination, \whp{},
the majority of calculated results for task $\tau$ are correct. 
Let us denote this event by $\mathcal{E}_{\tau}$.
It follows that
${\bf Pr}[\overline{\mathcal{E}}_\tau] \leq \frac{1}{n^{\beta}}$. 
Now, by Boole's inequality we obtain
\begin{displaymath}
{\bf Pr}[\bigcup_{\tau \in {\mathcal T}}\overline{\mathcal{E}}_\tau] \leq 
  \sum_{\tau\in \mathcal{T}} {\bf Pr}[\overline{\mathcal{E}}_\tau] \leq  
\frac{1}{n^{\beta-1}} \leq
\frac{1}{n^{\gamma}}
\end{displaymath}
where ${\mathcal T}$ is the set of all $n$ tasks, and $\gamma > 0$.

By  Lemmas~\ref{lem:halt},~\ref{lem:dissFa-p1} and~\ref{lem:dissFa-p2},
\whp{},  in $O(n^{1-a} \log{n}\log{\log{n}})$ rounds 
of the algorithm, at least $\Theta(n \log{n})$ triples generated by
processors in $P-F$ are
disseminated across all  workers \whp{}. Thus, the majority of the
results computed  for any task at any worker is the same  among all
 workers, and moreover these results are correct \whp{}. 
\end{proof}

According to Lemma~\ref{lem:dissFa-p2},
after $O(n^{1-\varepsilon}\log{n}\log{\log{n}})$ rounds of \epoch{b}
at least one processor in $P-F$ becomes \ENed{}. Furthermore, once a 
processor in $P-F$ becomes \ENed{}, according to Lemma~\ref{lem:halt}
after $O(\log{n})$ rounds of the algorithm every live processor 
becomes \ENed{} and then terminates, \whp{}. 
Next we assess  work and message complexities.

\begin{theorem} \label{complFa}
For $t=n$ algorithm \alg{} has work and message complexity $O(n\log n
\log{\log{n}})$. 
\end{theorem}

\begin{proof}
To obtain the result  we combine the costs associated with 
\epoch{a} with the costs of \epoch{b}.
The work and message complexity bounds for \epoch{a} are given by Lemma~\ref{lem:epocha}
as $\Theta(n \log n)$.

For \epoch{b} (if it is not empty), 
where $|P-F| = O(n^\varepsilon)$,  the algorithm terminates 
after $O(n^{1-\varepsilon}\log{n}\log{\log{n}})$
rounds \whp{} and there are $\Theta(n^{\varepsilon})$ live processors, thus its work is 
$O(n\log{n}\log{\log{n}})$. 
In every round if a processor is a {\em worker} it sends a {\sf share} message to one
randomly chosen processor.
If a processor is \ENed{} then it sends
{\sf profess} messages to a randomly selected subset of processors.
In every round 
$\Theta(n^\varepsilon)$ {\sf share} messages are sent. Since
the algorithm  terminates, \whp{}, in $O(n^{1-\varepsilon}\log{n}\log{\log{n}})$ rounds,
$\Theta(n\log{n}\log{\log{n}})$ {\sf share} messages are sent. 
On the other hand,
according to Lemma~\ref{lem:termin}, if during the execution of
the algorithm  $\Theta(n \log{n})$ {\sf profess} messages are sent then
every processor terminates \whp{}. Hence, the message complexity
 is $O(n\log{n}\log{\log{n}})$.

The worst case costs of the algorithm correspond to the executions
with non-empty \epoch{b}, where the algorithm does not terminate early. 
In this case the costs from  \epoch{a}
are asymptotically absorbed into the worst case costs of \epoch{b}
computed above.
\end{proof}

Finally, we consider the efficiency of algorithm \alg{_{t,n}} for $t$ tasks,
where $t\ge n$.  Note that the only change in the algorithm is
that, instead of one task, processors perform chunks of $t/n$ tasks.
The communication pattern in the algorithm remains exactly the same.
The following 
result is directly obtained from the analysis of algorithm \alg{} for $t=n$
by multiplying the time and work complexities by
the size  $\Theta(t/n)$ of the chunk of tasks; the message complexity is unchanged.

\begin{theorem}\label{thm:mainFa}
Algorithm \alg{_{t,n}}, with $t\ge n$,
computes the results of $t$ tasks correctly in model \M{fp} \whp{},
with time complexity
$O(\frac{t}{n^{\varepsilon}}\log{n}\log{\log{n}})$, work complexity
$O(t\log{n} \log{\log{n}})$, and message
complexity $O(n\log{n}\log{\log{n}})$.
\end{theorem}

\begin{proof}
For \epoch{a} algorithm \alg{} has time $\Theta(\log n)$, work $\Theta(t \log n)$,
and message complexity is $\Theta(n \log n)$.
The same holds for algorithm \alg{_{t,n}}.
For \epoch{b} algorithm \alg{}
takes $O(n^{1-\varepsilon}\log n
\log{\log{n}})$ iterations for at least one processor from set $P-F$ to
become \ENed{} \whp{}.
The same holds for \alg{_{t,n}}, except that each iteration 
is extended by $\Theta(t/n)$ rounds due to the size of  chunks
(recall that no communication takes place during these rounds). 
This yields round complexity $O(\frac{t}{n^{\varepsilon}}\log{n}\log{\log{n}})$.
Work complexity is then $O(t\log{n} \log{\log{n}})$.
Message complexity remains the same as for algorithm \alg{}
at $O(n\log{n}\log{\log{n}})$ as 
the number of messages does not change.
The final assessment is obtained by combining the costs
of \epoch{a}  and \epoch{b}.
\end{proof}


\subsection{Failure Model \M{pl}}
We start with the analysis of algorithm \alg{}, then
extend the main result to algorithm \alg{_{t,n}}, for
 the adversarial  model \M{pl}, where  $|P-F| = \Omega(\polylog{n})$,
where we use the  term $\polylog{n}$ to denote a member of the class of functions
$\bigcup_{k\geq 1}O(\log^k{n})$.
As a motivation,  first note that when a large number of crashes make $|P-F| =
\Theta(\polylog{n})$,
one may attempt a trivial solution where all live processors perform all $t$ tasks.
While this approach has efficient work, it does not guarantee that 
workers compute correct results; in fact, since the overall probability of live
workers producing bogus results can be close to $\frac{1}{2}$, this 
may yield on the average just slightly more than $t/2$ correct results.

For executions in \M{pl}, let $|P-F|$ be at least $a \log^c n$, 
for specific constants $a$ and $c$ satisfying the model constraints.
Let $F_r$ be the actual number of crashes that occur
prior to round $r$. 
For the purpose of analysis we divide an execution
of the algorithm into two epochs: \epoch{b'} and \epoch{c}.
In \epoch{b'} we include all rounds $r$ where $F_r$ 
remains constrained as in model \M{fp}, i.e.,
$|P-F_r| \ge b n^\varepsilon$, for some constants $b>0$ and $\varepsilon\in (0,1)$;
for reference, this epoch combines \epoch{a} and \epoch{b}
from the previous section.
In \epoch{c} we include all rounds $r$ starting with 
the first round $r''$ (it can be round 1) when the number of live
processors drops below $b n^\varepsilon$, 
but  remains $\Omega(\log^c{n})$ per model \M{p\ell}.%
%
Also note that either epoch may be empty.

In \epoch{b'} the algorithm incurs costs exactly as
in model \M{fp}.
Next we consider \epoch{c}.
If  algorithm \alg{} terminates in round $r''$, the first round of the epoch,
the costs remain the same as the costs analyzed for \M{fp}
in the previous section.

If it does not terminate, it incurs additional
costs associated with the processors in $P - F_{r''}$, where 
$a \log^c n \le |P - F_{r''}| \le b n^\varepsilon$.
We analyze the costs for \epoch{c} next. 
The final message and work complexities
are then at most the worst case complexity for \epoch{b'}
plus the additional costs for \epoch{c}.

In the next lemmas we use the fact that $|P - F_{r''}|=\Omega(\log^c{n})$.
The first lemma shows that within some $O(n)$ rounds in \epoch{c}
every task is chosen for execution
 $\Theta(\log{n})$ times by processors in $P-F$ \whp{}.

\begin{lemma} \label{thresholdFb}
In $O(n)$ rounds of \epoch{c} every task is performed
$\Theta(\log{n})$ times \whp{}  by  processors in $P-F$.
\end{lemma}

\begin{proof}
{If the algorithm terminates within 
$O(n)$
rounds of \epoch{c}, then each task is performed
$\Theta(\log n)$ times as reasoned earlier.
Suppose the algorithm does not terminate
(its performance is worse in this case).}
Let us assume that after $\tilde{r}$ rounds of algorithm \alg{}, where
$\tilde{r} = \tilde{\kappa} n$  ($\tilde{\kappa}$ is a sufficiently large constant),
there exists a task $\tau$ that is performed less than $(1-\delta)\tilde{\kappa} \log{n}$ times 
by the processors in $P-F$, for some $\delta > 0$.  
We prove that \whp{} such a task does not exist.

We define $k_3$ to be such that $k_3 = (1-\delta) \tilde{\kappa}$
(the constant $k_3$ will play a role in establishing the value of the compile-time
  constant $\const{K}$ of algorithm \alg{}; we come
  back to this at the end of Section~\ref{analysis}).
According to the above assumption,
 at the end of round $\tilde{r}$ for some task $\tau$, we have 
$|\cup_{j=1}^{n} R_j[\tau]| < (1-\delta) \tilde{\kappa} \log n = k_3 \log{n}$.

Let us consider all algorithm iterations individually performed by each processor in $P-F$
during the $\tilde{r}$ rounds.
Let $\xi$ be the total number of such individual iterations. 
Then $\xi \ge \tilde{r} |P-F| \geq \tilde{r}a \log^c{n}$. 
During any such iteration, a processor from $P-F$ selects and performs 
task $\tau$ in  line~\algoref{algo:choose_task} independently
with probability $\frac{1}{n}$.
Let us arbitrarily enumerate said iterations from $1$ to $\xi$.
Let $X_1, \ldots, X_x, \dots, X_\xi$ be Bernoulli random variables, such that
$X_x$ is $1$ if task  $\tau$ is performed in iteration $x$, and $0$ otherwise.
We define $X \equiv  \sum_{x=1}^{\xi}X_x$, the random variable that describes
the total number of times task $\tau$ is performed during the 
$\tilde{r}$ rounds by processors in $P-F$. 
We define $\mu$ to be ${\mathbb E}[X]$.
Since ${\bf Pr}[X_x=1] = \frac{1}{n}$, for $x\in \{1, \ldots, \xi\}$,
where $\xi \geq \tilde{r}a \log^c{n}$,
by linearity of expectation, we obtain
  $\mu = {\mathbb E}[X] = \sum_{x=1}^{\xi} \frac{1}{n} 
 \geq \tilde{\kappa} a\log^c{n}> k_3 \log{n}$.
Now by applying Chernoff bound  for the same $\delta>0$ as chosen
above, we have: 

\begin{displaymath}
{\bf Pr}[ X \leq (1-\delta)\mu] \leq e^{ -\frac{\mu\delta^2}{2}} 
\leq e^{-\frac{(\tilde{\kappa} a\log^c{n})\delta^2}{2}} \leq \frac{1}{n^{\frac{k\log^{c-1}n\delta^2}{2}}} \leq \frac{1}{n^{\alpha}}
\end{displaymath}
where $\alpha >1$ for some sufficiently large $k$.
Now let  $\mathcal{E}_\tau$ denote the probability event  
 $|\cup_{i=1}^{n} R_i(\tau)| > k_3 \log{n}$  by the round $\tilde{r} $ of the algorithm, and 
we let $\bar{\mathcal{E}_\tau}$ be the complement of that event.  
Next, by Boole's inequality we have 
${\bf Pr}[ \cup_\tau \bar{\mathcal{E}}_\tau ] \leq  \sum_\tau{\bf Pr}[\bar{\mathcal{E}_\tau} ] \leq \frac{1}{n^{\beta}}$,
where $\beta =\alpha - 1 > 0$. Hence each task  is performed at least
$\Theta(\log{n})$  times \whp{}, i.e., 
${\bf Pr}[ \cap_\tau \mathcal{E}_\tau ] = 
{\bf Pr}[ \overline{\cup_\tau\bar{\mathcal{E}_\tau}} ] \geq  1 -
\frac{1}{n^{\beta}}$. 
\end{proof}

Next we show that once each task is done a logarithmic number of times
by processors in $P-F$,
then at least one worker in $P-F$ acquires a sufficient collection of triples in
at most a linear number of rounds to become enlightened.
{We note that
if the algorithm terminates within 
$O(n)$ rounds of \epoch{c}, then every processor in $P-F$ is enlightened as
reasoned earlier. Suppose the algorithm does not terminate
(leading to its worst case performance).}

\begin{lemma} \label{lem:dissFb}
Once every task is performed $\Theta(\log{n})$ times by processors in $P-F$
then at least one worker in $P-F$ becomes \ENed{}
\whp{} after  $O(n)$ rounds in \epoch{c}. 
\end{lemma}

\begin{proof}
Assume that after $r$ rounds of algorithm \alg{}, every task $j \in \cal{T}$
is done $\Theta(\log{n})$ times by processors in $P-F$, and let $\mathcal{V}$ be the set 
of corresponding triples in the system.
Consider a triple $z \in \cal{V}$ that was generated in some round $\tilde{r}$.
 We want to prove that \whp{} it  takes $O(n)$ rounds for
the rest of the processors in $P-F$ to learn about $z$.

Let $\Lambda(n)$ be the number of processors in $P-F$, then
$|P-F| = \Lambda(n) \geq a \log^c{n}$, by the constraint of  model \M{pl}.
While there may be more than $\Lambda(n)$ processors
that start \epoch{c}, we focus only on processors in $P-F$. 
This is sufficient for our purpose of establishing
an upper bound on the number of rounds of at least one worker becoming
enlightened:
in line~\algoref{alg:share} of algorithm \alg{} every live
processor chooses a destination for a {\sf share} message uniformly
at random, and hence having more processors will only cause
a processor in $P-F$ becoming enlightened quicker.

Let $Z(r) \subseteq {P-F}$ be the set of processors that becomes aware of 
triple $z$, in round $r$.
Beginning with round $\tilde{r}$ when the triple is generated,  
we have $|Z(\tilde{r})| \geq 1$ (at least one processor is aware of the triple).
For any  rounds $r'$ and $r''$, where $\tilde{r}\le r' \le r''$,
we have $Z(\tilde{r})\subseteq Z(r') \subseteq Z(r'') \subseteq P-F$
because the considered processors that become aware of $z$ do not crash;
thus  $|Z(r)|$ is monotonically non-decreasing  with respect to $r$.

We want to estimate an upper bound on the total number of rounds $r$ required for
$|Z(r)|$ to become $\Lambda(n)$. 
We will do this by constructing a sequence of random mutually independent variables,
each corresponding to a contiguous segment of 
rounds $r_1, ... ,r_k$, for $k \geq 1$ in an execution of the algorithm.
Let $r_0$ be the round that precedes round $r_1$.
Our contiguous segment of rounds has the following properties:
(a) $|Z(r_x)|=|Z(r_0)|$ for $1\le x <k$, where during
such rounds $r_x$ the set $Z(r_x)$ does not grow
(the set of such rounds may be empty), and
(b) $|Z(r_k)| > |Z(r_0)|$, i.e., the size of the set grows.

For the purposes of analysis, we assume
that $|Z(r_k)| = |Z(r_0)|+1$, i.e., the set grows
by exactly one processor. Of course it is possible
that this set grows by more than a unity in one round.
Thus we consider an `amnesiac' version of the algorithm
where if more than one processor learns about the triple,
then all but one processor `forget' about that triple.
The information is propagated slower in the amnesiac
algorithm, but this is sufficient for us to establish
the needed upper bound on the number of rounds needed
to propagate the triple in question.

Consider some round $r$ with $|Z(r)| = \lambda$.
We define random variable $T_\lambda$  that represents the
 number of rounds required for $|Z(r+T_\lambda)| =\lambda+1$, i.e.,
$T_\lambda$ corresponds to the number $k$ of rounds in the contiguous segment
of rounds we defined above.
The random variables $T_\lambda$ are geometric, independent random
variables. Hence, we acquire a
sequence of random variables 
$T_1, ... , T_{\Lambda(n) -1}$, 
since $|P-F| = \Lambda(n)$ and according to our amnesiac algorithm
$|Z(r)| \le |Z(r+1)| +1$ for any round $r\ge\tilde{r}$.

Let us define the random variable  $T$ as
$T \equiv \sum_{\lambda=1}^{\Lambda(n)-1}T_\lambda$.
$T$ is  the total number of rounds required for all processors 
in $P-F$ to learn about triple~$z$:
By Markov's inequality we have:
\begin{displaymath}
{\bf Pr}(\kappa T > \kappa \eta) = {\bf Pr}(e^{\kappa T}>e^{k\eta})
 \leq \frac{\mathbb{E}[e^{\kappa T}]}{e^{\kappa \eta}} \; ,
\end{displaymath}
for some $\kappa>0$ and $\eta > 0$ to be specified later in the proof.

We say that ``a transmission in round $r > \tilde{r}$ successful" if  processor $j \in Z(r)$ sends a
message to some processor $l \in (P-F) - Z(r)$; otherwise we say that
``the transmission is unsuccessful."
Let $p_j$ be the probability that the transmission is successful in a round, and
$q_j=1-p_j$ be the probability that it is unsuccessful.  Note that if
a transmission is unsuccessful 
then this means that in that round none of the processors in $Z(r)$,
where $|Z(r)|=\lambda$, were able to contact a processor in $(P-F) -
Z(r)$ (here $|(P-F) - Z(r)|=\Lambda(n) - \lambda$), and hence we have:
\begin{displaymath}
q_j = (1-\frac{\Lambda(n) - \lambda}{n})^\lambda
\end{displaymath}
By geometric distribution, we have the following:
\begin{eqnarray*}
\mathbb{E}[e^{\kappa T_\lambda}] = p_je^\kappa + p_je^{2\kappa}q_j + p_j e^{3\kappa} {q_j}^2 + ...  =
p_je^\kappa(1+q_je^\kappa+{q_j}^2e^{2\kappa} + ...)
\end{eqnarray*}

{In order to sum the infinite geometric
  series, we need to have $q_je^\kappa < 1$.} Assume that $q_je^\kappa
< 1$ (note that we will need to choose $\kappa$ such 
that the inequality is satisfied), hence using infinite geometric
series we have:
\begin{eqnarray*}
\mathbb{E}[e^{\kappa T_\lambda}] = \frac{p_je^\kappa}{1-q_je^\kappa}
\end{eqnarray*}
In the remainder of the proof we focus on deriving a tight bound on the 
$\mathbb{E}[e^{\kappa T_\lambda}]$, and subsequently apply Boole's inequality across
all triples in ${\mathcal V}$

\begin{eqnarray*}
{\bf Pr}[\kappa T \geq \kappa \eta] &\leq& \frac{\mathbb{E}[e^{\kappa T}]}{e^{\kappa \eta}} =
\frac{\mathbb{E}[e^{\kappa \sum_{\lambda=1}^{\Lambda(n)-1}T_\lambda}]}{e^{\kappa \eta}} =
\frac{\prod_{\lambda=1}^{\Lambda(n)-1} \mathbb{E}[e^{\kappa T_\lambda}]}{e^{\kappa \eta}} \\
&=&\frac{1}{e^{\kappa \eta}} \prod_{\lambda=1}^{\Lambda(n)-1}
\frac{(1-(1-\frac{\Lambda(n)-\lambda}{n})^\lambda)e^\kappa}{1-(1-\frac{\Lambda(n)-\lambda}{n})^\lambda
  e^\kappa}
\\
&\leq& \frac{1}{e^{\kappa \eta}} \prod_{\lambda=1}^{\Lambda(n)-1}
\frac{(1-1+\frac{\lambda}{n}(\Lambda(n)-\lambda))e^\kappa}{1-(1-\frac{\Lambda(n)-\lambda}{n})^\lambda
  e^\kappa}
\\
&=& \frac{1}{e^{\kappa \eta}} \prod_{\lambda=1}^{\Lambda(n)-1}
\frac{\lambda
  (\Lambda(n)-\lambda)e^\kappa}{n(1-(1-\frac{\Lambda(n)-\lambda}{n})^\lambda
  e^\kappa)}
\end{eqnarray*}
Remember that we assumed that $e^\kappa q_j=e^\kappa(1-\frac{\Lambda(n)-\lambda}{n})^\lambda < 1$, 
for $j \in Z(r)$, $\lambda=1,2,\cdots, \Lambda(n)-1$ and $\kappa>0$. Let $\kappa$
be such that $e^\kappa = 1+\frac{\Lambda(n) }{2n}$, then we have the
following
\begin{eqnarray*}
e^\kappa(1-\frac{\Lambda(n) -\lambda}{n})^\lambda &=&
(1+\frac{\Lambda(n)}{2n})(1-\frac{\Lambda(n) -\lambda}{n})^\lambda \\
&=& 1+\frac{\Lambda(n)}{2n} - \frac{(\Lambda(n) -\lambda)\lambda}{n} -
O(\frac{1}{n^2}) \\
&=& 1 - \frac{\lambda (\Lambda(n) -\lambda) - \frac{1}{2}\Lambda(n)}{n} - O(\frac{1}{n^2}) 
\end{eqnarray*}

In order to show that $e^\kappa(1-\frac{\Lambda(n) -\lambda}{n})^\lambda <1$ it
remains to show that $\lambda (\Lambda(n) -\lambda) - \frac{1}{2}\Lambda(n)$
is positive. Note that $\lambda (\Lambda(n) -\lambda)$ is increasing until $\lambda
\leq \frac{\Lambda(n)}{2}$, we should also note that we consider
cases for  $\lambda=1,2,\cdots,\Lambda(n)-1$. Hence, the minimal value of
$\lambda(\Lambda(n) -\lambda)$ will be when either $\lambda=1$, or
$\lambda=\Lambda(n) -1$ and in both cases $\lambda (\Lambda(n) -\lambda) -
\frac{1}{2}\Lambda(n) \geq 0$, for sufficiently large $n$. 

Let us now evaluate the following expression:

$\displaystyle
n(1-e^\kappa(1-\frac{\Lambda(n)-\lambda}{n})^\lambda) 
$
\begin{eqnarray*}
&=&
n(1-(1-\frac{\lambda(\Lambda(n)-\lambda) - \frac{1}{2}\Lambda(n)}{n} -
O(\frac{1}{n^2}))) \\
&=& n(\frac{\lambda (\Lambda(n)-\lambda) - \frac{1}{2}\Lambda(n)}{n} +
O(\frac{1}{n^2})) \\
&=& \lambda (\Lambda(n)-\lambda) - \frac{1}{2}\Lambda(n) + O(\frac{1}{n})
\end{eqnarray*}

Then, we have

$\displaystyle
\prod_{\lambda=1}^{\Lambda(n)-1}
\frac{\lambda (\Lambda(n)-\lambda)}{n(1-e^\kappa(1-\frac{\Lambda(n)-\lambda}{n})^\lambda)}
$
\begin{eqnarray*}
&\leq&
\prod_{\lambda=1}^{\Lambda(n)-1}\frac{\lambda (\Lambda(n)-\lambda)}{\lambda (\Lambda(n)-\lambda)
- \frac{1}{2}\Lambda(n)+O(\frac{1}{n})} \\
&\leq&
\prod_{\lambda=1}^{\Lambda(n)-1}(1-\frac{\Lambda(n)}{2 \lambda (\Lambda(n)-\lambda)})^{-1} \\
&\leq& \prod_{\lambda=1}^{\Lambda(n)-1} \left(\frac{1}{2}\right)^{-1} \leq 2^{\Lambda(n)}
\end{eqnarray*}
The latest is true because $\lambda (\Lambda(n)-\lambda)$ achieves its minimal
value when $\lambda=1$.
Now, since $\Lambda(n) \geq a \log^c{n}$  we have:
\begin{eqnarray*}
{\bf Pr}(\kappa T > \kappa \eta) \leq \frac{e^{\kappa(a \log^c{n}-1)}}{e^{\kappa \eta}}
2^{a \log^c{n}}
\end{eqnarray*}

Since $e^\kappa = 1+ \frac{a \log^c{n}}{2n}$ then by taking natural base
logarithm of both sides and using Taylor series for $\ln(1+x)$,
where $|x|<1$, we have $\kappa \leq \frac{a \log^c{n}}{2n}$. And hence,
$e^{\kappa(a \log^c{n}-1)}=O(1)$. 
And we get,
\begin{eqnarray*}
{\bf Pr}(\kappa T > \kappa \eta) \leq \frac{2^{a \log^c{n}}}{e^{\kappa \eta}}
\end{eqnarray*}
By taking $\eta = kn$, where $k>2$ is a sufficiently large constant, we get
\begin{eqnarray*}
{\bf Pr}(\kappa T > \kappa \eta) \leq
\frac{2^{a \log^c{n}}}{e^{\frac{a \log^c{n} kn}{2n}}} \leq \frac{1}{n^{\alpha}}
\end{eqnarray*}
where $\alpha > 1$ for some sufficiently large constant $k>2$.

Thus we showed that if a new triple is generated by a worker in $P-F$ then \whp{} it is
known to all processors in $P-F$ in $O(n)$ rounds.
Now by applying Boole's inequality we want to show that \whp{} in
$O(n)$ rounds all triples in $\cal{V}$ become known to all processors in $P-F$.

Let $\overline{\mathcal{E}}_z$ be the event that some triple 
$z \in \cal{V}$ is not spread around among all workers in
$P-F$. In the preceding part of the proof  we have 
shown that ${\bf Pr}[\overline{\mathcal{E}}_z] <
\frac{1}{n^\alpha}$, where $\alpha>1$. By Boole's inequality, the
probability that there exists one triple that did not get spread to
all workers in $P-F$, can be bounded as  
\begin{displaymath}
{\bf Pr} [\cup_{z \in {\mathcal
    V}}\overline{\mathcal{E}}_z ] \leq \Sigma_{z \in
  {\mathcal V}}{\bf Pr}[\overline{\mathcal{E}}_z] = 
\Theta(n \log^c{n})\frac{1}{n^\alpha} \leq \frac{1}{n^\beta}
\end{displaymath}
where $\beta>0$. This implies that 
every  worker in $P-F$ collects all $\Theta(n \log n)$
triples generated by processors in $P-F$ \whp{}. Thus, at least one
worker in $P-F$ becomes \ENed{} after $O(n)$ rounds. 
\end{proof}

The following theorem shows that, with high probability, during \epoch{c} the 
correct results for all $n$ tasks, are  available at all live processors in $O(n)$ rounds.

\begin{theorem}\label{thm:correct:resFb}
Algorithm \alg{} makes known the correct results of all $n$ tasks
at every live processor in \epoch{c}
after $O(n)$ rounds \whp{}.
\end{theorem}

\begin{proofsketch}
The proof of this theorem is similar to the proof of
Theorem~\ref{thm:correct:resFa}. This is because, by
Lemma~\ref{thresholdFb}, in $O(n)$ rounds
the processors in $P-F$ generate
$\Theta(\log^c{n})$ triples, where $c \geq 1$ is a constant.
According to
Lemmas~\ref{lem:halt} and~\ref{lem:dissFb} in $O(n)$ rounds 
every live worker  becomes \ENed{}.
\end{proofsketch}

According to Lemma~\ref{lem:dissFb},
after $O(n)$ rounds of \epoch{c}
at least one processor in $P-F$ becomes \ENed{}. Furthermore, once a 
processor in $P-F$ becomes \ENed{}, according to Lemma~\ref{lem:halt}
after additional $O(\log{n})$ rounds  every live processor 
becomes \ENed{} and then terminates, \whp{}. 
Next we assess work and message complexities
(using the approach in the proof of Theorem~\ref{complFa}).
Recall  that we may choose arbitrary $\varepsilon$,
such that $0<\varepsilon < 1$.

\begin{theorem} \label{complFc}
Under adversarial model \M{pl} algorithm \alg{}  has work complexity and message complexity
$O(n^{1+\varepsilon})$, for any $0<\varepsilon < 1$. 
\end{theorem}

\begin{proof}
To obtain the result  we combine the costs associated with 
\epoch{b'} with the costs of \epoch{c}.
As reasoned earlier, the worst case costs for \epoch{b'}
are given in Theorem~\ref{complFa}.


{
For \epoch{c} (if it is not empty), 
where $|P-F| = \Omega(\log^cn)$, algorithm \alg{} terminates 
after $O(n)$ rounds \whp{} and 
there are up to $O(n^\varepsilon)$ live processors.
Thus its work is 
$O(n)\cdot O(n^\varepsilon)=O(n^{1+\varepsilon})$. 
In every round, if a processor is a {\em worker} it sends a {\sf share} message to one
randomly chosen processor.
If a processor is \ENed{} then it sends
{\sf profess} messages to a randomly selected subset of processors.
In every round 
$O(n^{\varepsilon})$ {\sf share} messages are sent. Since
\whp{} algorithm \alg{} terminates in $O(n)$ rounds,
$O(n^{1+\varepsilon})$ {\sf share} messages are sent. On the other hand,
according to Lemma~\ref{lem:termin}, if during an execution 
$\Theta(n \log{n})$ {\sf profess} messages are sent then
every processor terminates \whp{}. Hence, the message complexity
is $O(n^{1+\varepsilon})$.

The worst case costs of the algorithm correspond to executions
with non-empty \epoch{c}, where the algorithm does not terminate early. 
In this case the costs from  \epoch{b'}
are asymptotically absorbed into the worst case costs of \epoch{c}
computed above.
}
\end{proof}

Last, we extend our analysis to assess the efficiency of algorithm \alg{_{t,n}} for $t$ tasks,
where $t\ge n$. 
This is done based on the definition of algorithm \alg{_{t,n}}
using the same observations as done in discussing Theorem~\ref{thm:mainFa}.

\begin{theorem}\label{thm:mainFc}
Algorithm \alg{_{t,n}}, with $t\ge n$,
computes the results of $t$ tasks correctly in adversarial model \M{pl} \whp{},
in $O(t)$ rounds, with work complexity $O(t \cdot n^\varepsilon)$ 
and message complexity $O(n^{1+\varepsilon})$, for any $0<\varepsilon < 1$.
\end{theorem}

\begin{proof}
The result for algorithm \alg{_{t,n}} is obtained (as in Theorem~\ref{thm:mainFa})
by combining the costs from \epoch{b'} (ibid.) with the costs
of \epoch{c} derived
from the analysis of algorithm \alg{} for $t=n$ (Theorem~\ref{complFc}).
This is done by multiplying the number of rounds and work complexities by
the size of the chunk $\Theta(t/n)$; the message complexity is unchanged.
\end{proof}

We note that it should be possible to derive tighter bounds
on the complexity of the algorithm. 
This is because we only assume that for all rounds in \epoch{c} the number of live processors
is bounded by the generous range
$a \log^c n \le |P - F_{r}| \le b n^\varepsilon$. 
In particular, if in all rounds of  \epoch{c} there are $\Theta(\polylog{n})$ live processors,
the round and message complexities  
both become $O(t \, \polylog{n})$
as follows from the arguments along the lines of the proofs of Theorems~\ref{complFc} and~\ref{thm:mainFc}.

\subsection{Finalizing Algorithm Parameterization}

Lastly, we discuss the compile-time constants $\const{H}$ and $\const{K}$ 
that appear in algorithm \alg{} (starting with line~\algoref{alg:const}).
Recall that we have already given the constant $\const{H}$ 
in Section~\ref{sec:lemmas}; the constant
stems from the proof of Lemma~\ref{lem:termin}.

We compute $\const{K}$ as $ \max\{k_1,k_2,k_3\}$, where $k_2$
  and $k_3$ come from the proofs of Lemmas~\ref{thresholdFa}
  and~\ref{thresholdFb}. 
The constant $k_1$, as we detail below, emerges from the proof  
of Lemma~2 of \cite{DKS2011} 
in the same way that the constants $k_2$ and $k_3$
are established in Lemmas~\ref{thresholdFa}
  and~\ref{thresholdFb}.

As we discussed in conjunction with Lemma~\ref{lem:epocha},
algorithm \alg{} in \epoch{a} performs tasks in the same pattern as in
algorithm $A$ \cite{DKS2011} when $\Omega(n)$ processors
do not crash. 
Lemma~2 of \cite{DKS2011} shows that after $\Theta(\log n)$
rounds of algorithm~$A$ there is no task that is performed
less than $k(1-\delta)\log n$ times, \whp{}, for 
a suitably large constant $k$ and some constant $\delta \in (0,1)$.
Thus, we let $k_1$ to be $k_1 = k(1-\delta)$.
This allows us to define $\const{K}$ to be $\const{K} = 
\max\{k_1,k_2,k_3\}$,
ensuring that the constant $\const{K}$ in algorithm \alg{} 
(and thus in algorithm \alg{_{t,n}}) is large enough
to satisfy all requirements of the analysis.

\section{Conclusion} \label{conclusion}

We presented a synchronous decentralized  algorithm
that can perform a set of tasks using a distributed system of 
undependable, crash-prone processors. Our
randomized algorithm allows the processors to compute the 
correct results and make the results available at every
live participating processor, \whp{}. We provided  time, message, and
work complexity bounds for two adversarial strategies, viz., (a) all but
$\Omega(n^\varepsilon)$, $\varepsilon > 0$, processors can crash,
and (b) all but a poly-logarithmic number of  processors can crash.  
{In this work our focus was on stronger adversarial behaviors,
while still assuming synchrony and reliable communication.}
Future work considers the problem in synchronous and
asynchronous decentralized systems, with more virulent adversarial
settings in both. 
We plan to derive strong lower bounds on the message,
time, and work complexities in various models.
{Our algorithm solves the problem (\whp{})
even if only one processor remains operational.
Thus it is worthwhile to understand its behavior
in light of failure dynamics during executions.
Accordingly we plan to derive complexity bounds
that depend on the number of processors and tasks,
and also on the actual number of crashes.

\bibliographystyle{plain}
\bibliography{byzantine}

\end{document}